%% file: main.tex
\documentclass[fleqn,orivec]{llncs}

\input{preamble}

\title{Bounded Cycle Synthesis\thanks{Partially supported by the DFG project ``AVACS'' (SFB/TR 14). The
second author was supported by an IMPRS-CS PhD Scholarship.}}

\author{Bernd Finkbeiner and Felix Klein}

\institute{Reactive Systems Group, Saarland University, Germany\\
\email{\{finkbeiner,fklein\}@cs.uni-saarland.de}}

\begin{document}
	
\maketitle

\begin{abstract}
  \input{abstract}
\end{abstract}

\section{Introduction}
\label{sec_introduction}
\input{introduction}

\section{Preliminaries}
\label{sec_preliminaries}
\input{preliminaries}

\section{Bounds on the number of cycles}
\label{sec_bounds}
\input{bounds}

\section{Bounding the Cycles}
\label{sec_encoding}
\input{encoding}

\section{Experimental Results}
\label{sec_results}
\input{results}

\section{Conclusions}
\label{sec_conclusions}
\input{conclusions}

\bibliographystyle{splncs03}
\bibliography{biblio}

\newpage
\appendix
\input{appendix}

\end{document}

%% file: preamble.tex
\usepackage[utf8]{inputenc}
\usepackage[T1]{fontenc}
\usepackage{wrapfig}
\usepackage{amsmath}
\usepackage{amssymb}
\usepackage{array}
\usepackage{multirow}
\usepackage{todonotes}
\usepackage{booktabs}
\usepackage{pifont}
\usepackage{tikz}
\usepackage{color}
\usepackage{float}
\usepackage{colortbl}
\usepackage{ltl}

\usetikzlibrary{calc}
\usetikzlibrary{automata}
\usetikzlibrary{backgrounds}
\usetikzlibrary{decorations.pathreplacing}

\newcommand{\cb}{\cellcolor{green!20}}

\newcommand{\inputs}{\ensuremath{\mathcal{I}}}
\newcommand{\outputs}{\ensuremath{\mathcal{O}}}
\newcommand{\mealy}{\ensuremath{\mathcal{M}}}
\newcommand{\sep}{\ensuremath{\ \ | \ \ }}
\newcommand{\nats}{\ensuremath{\mathbb{N}}}
\newcommand{\widx}[2]{\ensuremath{#1_{#2}}}
\newcommand{\set}[1]{\ensuremath{\{#1\}}}
\newcommand{\lang}{\ensuremath{\mathcal{L}}} 
\newcommand{\size}[1]{|#1|}
\newcommand{\aut}{\mathcal{A}}
\newcommand{\bigo}{\!\mathit{\raisebox{-0.5pt}{\textit{O}}}}
\newcommand{\trans}[3]{\ensuremath{\textsc{trans}(#1,#2,#3)}} 
\newcommand{\rgstate}[2]{\ensuremath{\textsc{rgstate}(#1,#2)}}
\newcommand{\tlabel}[3]{\ensuremath{\textsc{label}(#1,#2,#3)}}  
\newcommand{\annotationd}[3]{\ensuremath{\textsc{annotation}(#1,#2,#3)}}  
\newcommand{\annotation}[2]{\ensuremath{\textsc{annotation}(#1,#2)}}  
\newcommand{\edge}[2]{\ensuremath{\textsc{edge}(#1,#2)}}  
\newcommand{\redge}[2]{\ensuremath{\textsc{redge}(#1,#2)}}  
\newcommand{\bedge}[2]{\ensuremath{\textsc{bedge}(#1,#2)}}  
\newcommand{\wtreed}[2]{\ensuremath{\textsc{wtree}(#1,#2)}}  
\newcommand{\wtree}[1]{\ensuremath{\textsc{wtree}(#1)}}  
\newcommand{\rboundd}[2]{\ensuremath{\textsc{rbound}(#1,#2)}}  
\newcommand{\rbound}[1]{\ensuremath{\textsc{rbound}(#1)}}  
\newcommand{\forwardedge}[3]{\ensuremath{\textsc{forward}(#1,#2,#3)}}  
\newcommand{\backwardedge}[3]{\ensuremath{\textsc{backward}(#1,#2,#3)}}  
\newcommand{\frankd}[3]{\ensuremath{\textsc{frank}(#1,#2,#3)}}  
\newcommand{\brankd}[3]{\ensuremath{\textsc{brank}(#1,#2,#3)}}  
\newcommand{\frank}[2]{\ensuremath{\textsc{frank}(#1,#2)}}  
\newcommand{\brank}[2]{\ensuremath{\textsc{brank}(#1,#2)}}  
\newcommand{\allowed}[2]{\ensuremath{\textsc{visited}(#1,#2)}}  
\newcommand{\tree}{\ensuremath{\mathcal{T}}} 
\newcommand{\sccd}[3]{\ensuremath{\textsc{scc}(#1,#2,#3)}}  
\newcommand{\scc}[2]{\ensuremath{\textsc{scc}(#1,#2)}}  
\newcommand{\curlyF}{\ensuremath{\mathcal{F}}}

%% file: abstract.tex
We introduce a new approach for the
synthesis of Mealy machines from specifications in linear-time
temporal logic (LTL), where the number of cycles in the state graph
of the implementation is limited by a given bound.
Bounding the number of
cycles leads to implementations that are structurally simpler and
easier to understand.
We solve the synthesis problem via an extension of
SAT-based bounded synthesis, where we additionally construct a witness structure that limits the number of cycles.
We also establish a triple-exponential upper and lower bound
for the potential blow-up between the length of the LTL formula and the number of
cycles in the state graph.

%
%

%% file: introduction.tex
There has been a lot of recent progress in the automatic synthesis of
reactive systems from specifications in temporal logic~\cite{Bloem:2009,Filiot:2011,Filiot:2013,Schewe:2013,Jobstmann:2007}.
From a theoretical point of view, the appeal of synthesis is obvious: 
the synthesized implementation
is guaranteed to satisfy the specification. No separate
verification is needed.

From a practical point of view, the value proposition is not so clear.
Instead of writing programs, the user of a synthesis procedure now
writes specifications.  But many people find it much
easier to understand the precise meaning of a program than to
understand the precise meaning of a temporal formula. Is it really
justified to place higher trust into a program that was synthesized
automatically, albeit from a possibly ill-understood specification,
than in a manually written, but well-understood program?
A straightforward solution would be for the programmer to \emph{inspect} the
synthesized program and confirm that the implementation is indeed as intended. However, current synthesis tools fail miserably
at producing readable code.

Most research on the synthesis problem has focused on the problem of finding 
\emph{some} implementation, not necessarily a high-quality implementation.
Since specification languages like LTL restrict the behavior of a
system, but not its structure, it is no surprise that the synthesized
implementations are often much larger and much more complex than
a manual implementation. There has been some progress on improving other quality
measures, such as the runtime performance~\cite{Bloem:2009}, but very little has been
done to optimize the \emph{structural quality} of the synthesized implementations (cf.~\cite{Kupferman:2012}).
Can we develop synthesis algorithms that produce implementations that
are small, structurally simple, and therefore easy to understand?

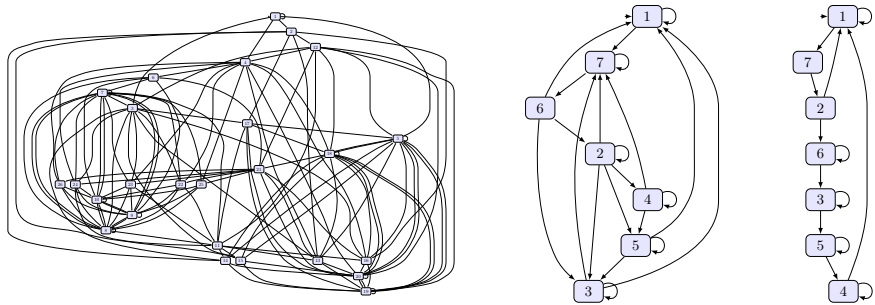
\begin{figure}[t]
\begin{tikzpicture}
  \node at (0,0) {
\scalebox{0.6}{\begin{tikzpicture}[>=latex,line join=bevel,scale=0.4]
\node (1) at (195.0bp,450.0bp) [draw,rounded corners=3,fill=blue!10,inner sep=4pt,minimum width=2em,initial text=,initial left] {1};
  \node (3) at (103.0bp,18.0bp) [draw,rounded corners=3,fill=blue!10,inner sep=4pt,minimum width=2em] {3};
  \node (2) at (121.0bp,234.0bp) [draw,rounded corners=3,fill=blue!10,inner sep=4pt,minimum width=2em] {2};
  \node (5) at (177.0bp,90.0bp) [draw,rounded corners=3,fill=blue!10,inner sep=4pt,minimum width=2em] {5};
  \node (4) at (196.0bp,162.0bp) [draw,rounded corners=3,fill=blue!10,inner sep=4pt,minimum width=2em] {4};
  \node (7) at (121.0bp,378.0bp) [draw,rounded corners=3,fill=blue!10,inner sep=4pt,minimum width=2em] {7};
  \node (6) at (27.0bp,306.0bp) [draw,rounded corners=3,fill=blue!10,inner sep=4pt,minimum width=2em] {6};
  \draw [->] (4) ..controls (228.69bp,177.68bp) and (241.0bp,173.53bp)  .. (241.0bp,162.0bp) .. controls (241.0bp,153.62bp) and (234.5bp,149.14bp)  ..  (4);
  \draw [->] (6) ..controls (58.95bp,281.21bp) and (79.223bp,266.11bp)  .. (2);
  \draw [->] (2) ..controls (116.44bp,178.83bp) and (109.07bp,91.181bp)  .. (3);
  \draw [->] (1) ..controls (227.69bp,465.68bp) and (240.0bp,461.53bp)  .. (240.0bp,450.0bp) .. controls (240.0bp,441.62bp) and (233.5bp,437.14bp)  ..  (1);
  \draw [->] (7) ..controls (89.05bp,353.21bp) and (68.777bp,338.11bp)  ..  (6);
  \draw [->] (2) ..controls (153.69bp,249.68bp) and (166.0bp,245.53bp)  .. (166.0bp,234.0bp) .. controls (166.0bp,225.62bp) and (159.5bp,221.14bp)  ..  (2);
  \draw [->] (7) ..controls (153.69bp,393.68bp) and (166.0bp,389.53bp)  .. (166.0bp,378.0bp) .. controls (166.0bp,369.62bp) and (159.5bp,365.14bp)  .. (7);
  \draw [->] (4) ..controls (189.42bp,198.07bp) and (183.13bp,227.43bp)  .. (175.0bp,252.0bp) .. controls (163.26bp,287.46bp) and (145.44bp,326.77bp)  ..  (7);
  \draw [->] (3) ..controls (183.11bp,36.763bp) and (307.0bp,74.912bp)  .. (307.0bp,161.0bp) .. controls (307.0bp,307.0bp) and (307.0bp,307.0bp)  .. (307.0bp,307.0bp) .. controls (307.0bp,361.47bp) and (255.56bp,407.33bp)  ..  (1);
  \draw [->] (5) ..controls (209.69bp,105.68bp) and (222.0bp,101.53bp)  .. (222.0bp,90.0bp) .. controls (222.0bp,81.622bp) and (215.5bp,77.143bp)  .. (5);
  \draw [->] (4) ..controls (189.25bp,136.14bp) and (186.65bp,126.54bp)  ..  (5);
  \draw [->] (1) ..controls (169.49bp,424.87bp) and (155.22bp,411.37bp)  ..  (7);
  \draw [->] (3) ..controls (91.025bp,74.636bp) and (73.66bp,171.65bp)  .. (85.0bp,252.0bp) .. controls (89.937bp,286.98bp) and (102.28bp,325.89bp)  ..  (7);
  \draw [->] (6) ..controls (43.602bp,342.85bp) and (61.452bp,375.3bp)  .. (85.0bp,396.0bp) .. controls (107.29bp,415.59bp) and (138.1bp,429.61bp)  ..  (1);
  \draw [->] (2) ..controls (121.0bp,276.42bp) and (121.0bp,320.89bp)  ..  (7);
  \draw [->] (3) ..controls (135.69bp,33.675bp) and (148.0bp,29.531bp)  .. (148.0bp,18.0bp) .. controls (148.0bp,9.6218bp) and (141.5bp,5.1433bp)  .. (3);
  \draw [->] (5) ..controls (151.49bp,64.872bp) and (137.22bp,51.369bp)  ..  (3);
  \draw [->] (5) ..controls (215.24bp,110.09bp) and (237.45bp,124.7bp)  .. (250.0bp,144.0bp) .. controls (272.05bp,177.91bp) and (269.0bp,192.55bp)  .. (269.0bp,233.0bp) .. controls (269.0bp,307.0bp) and (269.0bp,307.0bp)  .. (269.0bp,307.0bp) .. controls (269.0bp,353.41bp) and (237.53bp,399.67bp)  ..  (1);
  \draw [->] (2) ..controls (137.13bp,192.09bp) and (155.15bp,146.4bp)  .. (5);
  \draw [->] (6) ..controls (21.294bp,247.93bp) and (16.153bp,146.16bp)  .. (51.0bp,72.0bp) .. controls (57.056bp,59.112bp) and (67.568bp,47.478bp)  .. (3);
  \draw [->] (2) ..controls (146.85bp,208.87bp) and (161.32bp,195.37bp)  .. (4);
\end{tikzpicture}}
};

\node at (2.8,0) {
\scalebox{0.6}{\begin{tikzpicture}[>=latex,line join=bevel,scale=0.4]
\node (1) at (82.0bp,450.0bp) [draw,rounded corners=3,fill=blue!10,inner sep=4pt,minimum width=2em,initial text=,initial left] {1};
  \node (3) at (47.0bp,162.0bp) [draw,rounded corners=3,fill=blue!10,inner sep=4pt,minimum width=2em] {3};
  \node (2) at (47.0bp,306.0bp) [draw,rounded corners=3,fill=blue!10,inner sep=4pt,minimum width=2em] {2};
  \node (5) at (47.0bp,90.0bp) [draw,rounded corners=3,fill=blue!10,inner sep=4pt,minimum width=2em] {5};
  \node (4) at (83.0bp,18.0bp) [draw,rounded corners=3,fill=blue!10,inner sep=4pt,minimum width=2em] {4};
  \node (7) at (27.0bp,378.0bp) [draw,rounded corners=3,fill=blue!10,inner sep=4pt,minimum width=2em] {7};
  \node (6) at (47.0bp,234.0bp) [draw,rounded corners=3,fill=blue!10,inner sep=4pt,minimum width=2em] {6};
  \draw [->] (1) ..controls (114.69bp,465.68bp) and (127.0bp,461.53bp)  .. (127.0bp,450.0bp) .. controls (127.0bp,441.62bp) and (120.5bp,437.14bp)  ..  (1);
  \draw [->] (4) ..controls (115.69bp,33.675bp) and (128.0bp,29.531bp)  .. (128.0bp,18.0bp) .. controls (128.0bp,9.6218bp) and (121.5bp,5.1433bp)  ..  (4);
  \draw [->] (2) ..controls (55.81bp,334.42bp) and (59.887bp,347.9bp)  .. (63.0bp,360.0bp) .. controls (68.28bp,380.53bp) and (73.244bp,403.98bp)  ..  (1);
  \draw [->] (5) ..controls (79.688bp,105.68bp) and (92.0bp,101.53bp)  .. (92.0bp,90.0bp) .. controls (92.0bp,81.622bp) and (85.501bp,77.143bp)  ..  (5);
  \draw [->] (6) ..controls (79.688bp,249.68bp) and (92.0bp,245.53bp)  .. (92.0bp,234.0bp) .. controls (92.0bp,225.62bp) and (85.501bp,221.14bp)  ..  (6);
  \draw [->] (7) ..controls (34.102bp,352.14bp) and (36.846bp,342.54bp)  ..  (2);
  \draw [->] (1) ..controls (62.786bp,424.55bp) and (53.337bp,412.52bp)  ..  (7);
  \draw [->] (4) ..controls (101.09bp,61.849bp) and (120.0bp,114.38bp)  .. (120.0bp,161.0bp) .. controls (120.0bp,307.0bp) and (120.0bp,307.0bp)  .. (120.0bp,307.0bp) .. controls (120.0bp,348.59bp) and (104.6bp,394.78bp)  ..  (1);
  \draw [->] (3) ..controls (79.688bp,177.68bp) and (92.0bp,173.53bp)  .. (92.0bp,162.0bp) .. controls (92.0bp,153.62bp) and (85.501bp,149.14bp)  ..  (3);
  \draw [->] (2) ..controls (47.0bp,279.98bp) and (47.0bp,270.71bp)  ..  (6);
  \draw [->] (5) ..controls (59.712bp,64.283bp) and (65.147bp,53.714bp)  .. (4);
  \draw [->] (3) ..controls (47.0bp,135.98bp) and (47.0bp,126.71bp)  ..  (5);
  \draw [->] (6) ..controls (47.0bp,207.98bp) and (47.0bp,198.71bp)  ..  (3);
\end{tikzpicture}}
};

\node at (-5.2,0) {
\scalebox{0.2}{\begin{tikzpicture}[>=latex,line join=bevel,scale=0.4]
\node (24) at (318.0bp,522.0bp) [draw,rounded corners=3,fill=blue!10,inner sep=4pt,minimum width=2em] {24};
  \node (25) at (911.0bp,522.0bp) [draw,rounded corners=3,fill=blue!10,inner sep=4pt,minimum width=2em] {25};
  \node (26) at (246.0bp,522.0bp) [draw,rounded corners=3,fill=blue!10,inner sep=4pt,minimum width=2em] {26};
  \node (20) at (1651.0bp,90.0bp) [draw,rounded corners=3,fill=blue!10,inner sep=4pt,minimum width=2em] {20};
  \node (21) at (1183.0bp,594.0bp) [draw,rounded corners=3,fill=blue!10,inner sep=4pt,minimum width=2em] {21};
  \node (22) at (814.0bp,522.0bp) [draw,rounded corners=3,fill=blue!10,inner sep=4pt,minimum width=2em] {22};
  \node (23) at (578.0bp,522.0bp) [draw,rounded corners=3,fill=blue!10,inner sep=4pt,minimum width=2em] {23};
  \node (1) at (1260.0bp,1314.0bp) [draw,rounded corners=3,fill=blue!10,inner sep=4pt,minimum width=2em,initial text=,initial left] {1};
  \node (3) at (587.0bp,882.0bp) [draw,rounded corners=3,fill=blue!10,inner sep=4pt,minimum width=2em] {3};
  \node (2) at (1335.0bp,1242.0bp) [draw,rounded corners=3,fill=blue!10,inner sep=4pt,minimum width=2em] {2};
  \node (5) at (1838.0bp,738.0bp) [draw,rounded corners=3,fill=blue!10,inner sep=4pt,minimum width=2em] {5};
  \node (4) at (1118.0bp,1098.0bp) [draw,rounded corners=3,fill=blue!10,inner sep=4pt,minimum width=2em] {4};
  \node (7) at (445.0bp,954.0bp) [draw,rounded corners=3,fill=blue!10,inner sep=4pt,minimum width=2em] {7};
  \node (6) at (685.0bp,1026.0bp) [draw,rounded corners=3,fill=blue!10,inner sep=4pt,minimum width=2em] {6};
  \node (9) at (585.0bp,378.0bp) [draw,rounded corners=3,fill=blue!10,inner sep=4pt,minimum width=2em] {9};
  \node (8) at (462.0bp,306.0bp) [draw,rounded corners=3,fill=blue!10,inner sep=4pt,minimum width=2em] {8};
  \node (11) at (987.0bp,234.0bp) [draw,rounded corners=3,fill=blue!10,inner sep=4pt,minimum width=2em] {11};
  \node (10) at (419.0bp,450.0bp) [draw,rounded corners=3,fill=blue!10,inner sep=4pt,minimum width=2em] {10};
  \node (13) at (1460.0bp,162.0bp) [draw,rounded corners=3,fill=blue!10,inner sep=4pt,minimum width=2em] {13};
  \node (12) at (1449.0bp,1170.0bp) [draw,rounded corners=3,fill=blue!10,inner sep=4pt,minimum width=2em] {12};
  \node (15) at (1097.0bp,162.0bp) [draw,rounded corners=3,fill=blue!10,inner sep=4pt,minimum width=2em] {15};
  \node (14) at (1025.0bp,162.0bp) [draw,rounded corners=3,fill=blue!10,inner sep=4pt,minimum width=2em] {14};
  \node (17) at (1128.0bp,810.0bp) [draw,rounded corners=3,fill=blue!10,inner sep=4pt,minimum width=2em] {17};
  \node (16) at (1688.0bp,162.0bp) [draw,rounded corners=3,fill=blue!10,inner sep=4pt,minimum width=2em] {16};
  \node (19) at (1687.0bp,18.0bp) [draw,rounded corners=3,fill=blue!10,inner sep=4pt,minimum width=2em] {19};
  \node (18) at (1515.0bp,666.0bp) [draw,rounded corners=3,fill=blue!10,inner sep=4pt,minimum width=2em] {18};
  \draw [->] (25) ..controls (862.56bp,490.26bp) and (804.09bp,454.97bp)  .. (751.0bp,432.0bp) .. controls (706.97bp,412.95bp) and (653.78bp,397.16bp)  .. (9);
  \draw [->] (21) ..controls (1060.1bp,578.78bp) and (726.17bp,540.14bp)  .. (23);
  \draw [->] (2) ..controls (1295.7bp,1159.3bp) and (1178.8bp,916.58bp)  .. (17);
  \draw [->] (20) ..controls (1657.4bp,63.224bp) and (1662.9bp,52.281bp)  .. (19);
  \draw [->] (24) ..controls (324.32bp,485.38bp) and (331.68bp,455.07bp)  .. (345.0bp,432.0bp) .. controls (369.51bp,389.54bp) and (411.0bp,349.9bp)  ..  (8);
  \draw [->] (4) ..controls (971.27bp,1096.6bp) and (491.37bp,1092.2bp)  .. (350.0bp,1044.0bp) .. controls (295.4bp,1025.4bp) and (271.88bp,1021.3bp)  .. (242.0bp,972.0bp) .. controls (175.43bp,862.06bp) and (218.68bp,812.49bp)  .. (216.0bp,684.0bp) .. controls (215.0bp,635.94bp) and (209.91bp,622.53bp)  .. (222.0bp,576.0bp) .. controls (224.44bp,566.62bp) and (228.41bp,556.83bp)  .. (26);
  \draw [->] (3) ..controls (577.71bp,837.52bp) and (568.0bp,784.52bp)  .. (568.0bp,739.0bp) .. controls (568.0bp,739.0bp) and (568.0bp,739.0bp)  .. (568.0bp,665.0bp) .. controls (568.0bp,533.85bp) and (835.13bp,277.44bp)  .. (951.0bp,216.0bp) .. controls (996.45bp,191.9bp) and (1013.2bp,199.02bp)  .. (1061.0bp,180.0bp) .. controls (1062.8bp,179.27bp) and (1064.7bp,178.49bp)  .. (15);
  \draw [->] (6) ..controls (596.92bp,1017.1bp) and (448.76bp,1000.9bp)  .. (409.0bp,972.0bp) .. controls (372.48bp,945.45bp) and (360.0bp,928.15bp)  .. (360.0bp,883.0bp) .. controls (360.0bp,883.0bp) and (360.0bp,883.0bp)  .. (360.0bp,665.0bp) .. controls (360.0bp,594.83bp) and (389.84bp,516.23bp)  ..  (10);
  \draw [->] (26) ..controls (275.6bp,485.69bp) and (308.82bp,450.6bp)  .. (345.0bp,432.0bp) .. controls (411.29bp,397.92bp) and (498.86bp,385.67bp)  .. (9);
  \draw [->] (13) ..controls (1432.7bp,186.62bp) and (1418.1bp,201.08bp)  .. (1408.0bp,216.0bp) .. controls (1378.5bp,259.44bp) and (1382.9bp,276.75bp)  .. (1360.0bp,324.0bp) .. controls (1312.6bp,421.65bp) and (1314.1bp,455.49bp)  .. (1246.0bp,540.0bp) .. controls (1235.5bp,553.02bp) and (1221.4bp,565.21bp)  .. (21);
  \draw [->] (20) ..controls (1725.6bp,108.1bp) and (1831.9bp,143.32bp)  .. (1881.0bp,216.0bp) .. controls (1958.0bp,329.87bp) and (1932.0bp,383.56bp)  .. (1932.0bp,521.0bp) .. controls (1932.0bp,595.0bp) and (1932.0bp,595.0bp)  .. (1932.0bp,595.0bp) .. controls (1932.0bp,644.95bp) and (1892.8bp,692.54bp)  .. (5);
  \draw [->] (16) ..controls (1696.8bp,190.26bp) and (1700.4bp,203.73bp)  .. (1702.0bp,216.0bp) .. controls (1712.3bp,295.33bp) and (1712.3bp,316.66bp)  .. (1702.0bp,396.0bp) .. controls (1693.5bp,461.64bp) and (1699.5bp,484.11bp)  .. (1664.0bp,540.0bp) .. controls (1633.3bp,588.34bp) and (1577.9bp,627.48bp)  ..  (18);
  \draw [->] (2) ..controls (1139.9bp,1239.1bp) and (237.38bp,1228.4bp)  .. (122.0bp,1188.0bp) .. controls (58.657bp,1165.8bp) and (0.0000bp,1166.1bp)  .. (0.0bp,1099.0bp) .. controls (0.0bp,1099.0bp) and (0.0bp,1099.0bp)  .. (0.0bp,305.0bp) .. controls (0.0bp,203.17bp) and (804.03bp,170.33bp)  ..  (14);
  \draw [->] (12) ..controls (1535.7bp,1162.8bp) and (1684.8bp,1149.0bp)  .. (1805.0bp,1116.0bp) .. controls (1886.5bp,1093.6bp) and (1907.3bp,1085.4bp)  .. (1981.0bp,1044.0bp) .. controls (2039.2bp,1011.3bp) and (2102.0bp,1021.8bp)  .. (2102.0bp,955.0bp) .. controls (2102.0bp,955.0bp) and (2102.0bp,955.0bp)  .. (2102.0bp,161.0bp) .. controls (2102.0bp,111.22bp) and (2076.8bp,97.378bp)  .. (2034.0bp,72.0bp) .. controls (1982.1bp,41.214bp) and (1802.6bp,26.331bp)  .. (19);
  \draw [->] (26) ..controls (273.2bp,473.35bp) and (316.16bp,404.38bp)  .. (367.0bp,360.0bp) .. controls (385.95bp,343.46bp) and (410.9bp,329.78bp)  ..  (8);
  \draw [->] (7) ..controls (436.0bp,909.87bp) and (436.0bp,856.36bp)  .. (436.0bp,811.0bp) .. controls (436.0bp,811.0bp) and (436.0bp,811.0bp)  .. (436.0bp,593.0bp) .. controls (436.0bp,552.55bp) and (425.67bp,506.67bp)  .. (10);
  \draw [->] (5) ..controls (1871.0bp,697.0bp) and (1914.0bp,647.29bp)  .. (1914.0bp,595.0bp) .. controls (1914.0bp,595.0bp) and (1914.0bp,595.0bp)  .. (1914.0bp,521.0bp) .. controls (1914.0bp,383.56bp) and (1940.0bp,329.87bp)  .. (1863.0bp,216.0bp) .. controls (1819.9bp,152.24bp) and (1732.8bp,117.31bp)  ..  (20);
  \draw [->] (21) ..controls (1109.1bp,573.99bp) and (1001.4bp,546.27bp)  .. (25);
  \draw [->] (19) ..controls (1680.5bp,44.873bp) and (1675.1bp,55.8bp)  .. (20);
  \draw [->] (10) ..controls (467.77bp,425.37bp) and (518.42bp,403.99bp)  .. (9);
  \draw [->] (21) ..controls (1103.5bp,581.9bp) and (978.49bp,563.26bp)  .. (875.0bp,540.0bp) .. controls (866.21bp,538.02bp) and (856.8bp,535.57bp)  ..  (22);
  \draw [->] (24) ..controls (355.77bp,561.26bp) and (398.0bp,612.37bp)  .. (398.0bp,665.0bp) .. controls (398.0bp,739.0bp) and (398.0bp,739.0bp)  .. (398.0bp,739.0bp) .. controls (398.0bp,815.31bp) and (497.69bp,855.68bp)  ..  (3);
  \draw [->] (26) ..controls (229.94bp,579.04bp) and (207.12bp,678.44bp)  .. (232.0bp,756.0bp) .. controls (256.04bp,830.93bp) and (274.79bp,849.34bp)  .. (335.0bp,900.0bp) .. controls (357.71bp,919.11bp) and (388.48bp,933.17bp)  .. (7);
  \draw [->] (10) ..controls (440.42bp,493.12bp) and (454.0bp,546.63bp)  .. (454.0bp,593.0bp) .. controls (454.0bp,811.0bp) and (454.0bp,811.0bp)  .. (454.0bp,811.0bp) .. controls (454.0bp,851.22bp) and (454.0bp,897.85bp)  .. (7);
  \draw [->] (14) ..controls (1059.6bp,125.55bp) and (1100.2bp,88.681bp)  .. (1143.0bp,72.0bp) .. controls (1235.8bp,35.809bp) and (1542.5bp,23.327bp)  .. (19);
  \draw [->] (10) ..controls (424.47bp,407.19bp) and (438.07bp,361.45bp)  .. (8);
  \draw [->] (7) ..controls (477.69bp,969.68bp) and (490.0bp,965.53bp)  .. (490.0bp,954.0bp) .. controls (490.0bp,945.62bp) and (483.5bp,941.14bp)  .. (7);
  \draw [->] (21) ..controls (1024.4bp,591.57bp) and (453.56bp,583.8bp)  .. (282.0bp,540.0bp) .. controls (279.82bp,539.44bp) and (277.62bp,538.76bp)  ..  (26);
  \draw [->] (20) ..controls (1683.7bp,105.68bp) and (1696.0bp,101.53bp)  .. (1696.0bp,90.0bp) .. controls (1696.0bp,81.622bp) and (1689.5bp,77.143bp)  .. (20);
  \draw [->] (11) ..controls (1092.7bp,217.36bp) and (1332.2bp,181.92bp)  .. (13);
  \draw [->] (10) ..controls (451.69bp,465.68bp) and (464.0bp,461.53bp)  .. (464.0bp,450.0bp) .. controls (464.0bp,441.62bp) and (457.5bp,437.14bp)  .. (10);
  \draw [->] (24) ..controls (319.96bp,566.7bp) and (322.0bp,619.87bp)  .. (322.0bp,665.0bp) .. controls (322.0bp,811.0bp) and (322.0bp,811.0bp)  .. (322.0bp,811.0bp) .. controls (322.0bp,868.13bp) and (379.58bp,913.56bp)  ..  (7);
  \draw [->] (16) ..controls (1677.3bp,227.58bp) and (1650.8bp,365.32bp)  .. (1599.0bp,468.0bp) .. controls (1555.7bp,553.84bp) and (1530.9bp,567.06bp)  .. (1479.0bp,648.0bp) .. controls (1423.2bp,735.02bp) and (1312.6bp,975.69bp)  .. (1235.0bp,1044.0bp) .. controls (1211.2bp,1065.0bp) and (1177.7bp,1079.0bp)  .. (4);
  \draw [->] (8) ..controls (476.67bp,350.14bp) and (492.0bp,402.89bp)  .. (492.0bp,449.0bp) .. controls (492.0bp,523.0bp) and (492.0bp,523.0bp)  .. (492.0bp,523.0bp) .. controls (492.0bp,661.33bp) and (502.86bp,699.47bp)  .. (554.0bp,828.0bp) .. controls (557.99bp,838.02bp) and (563.86bp,848.25bp)  .. (3);
  \draw [->] (6) ..controls (617.72bp,1005.4bp) and (528.39bp,979.32bp)  .. (7);
  \draw [->] (11) ..controls (1000.4bp,208.28bp) and (1006.2bp,197.71bp)  .. (14);
  \draw [->] (1) ..controls (1285.9bp,1288.9bp) and (1300.3bp,1275.4bp)  .. (2);
  \draw [->] (16) ..controls (1751.3bp,213.31bp) and (1876.0bp,325.11bp)  .. (1876.0bp,449.0bp) .. controls (1876.0bp,595.0bp) and (1876.0bp,595.0bp)  .. (1876.0bp,595.0bp) .. controls (1876.0bp,636.59bp) and (1860.6bp,682.78bp)  .. (5);
  \draw [->] (3) ..controls (615.18bp,834.17bp) and (658.11bp,767.21bp)  .. (705.0bp,720.0bp) .. controls (706.03bp,718.96bp) and (1381.8bp,216.85bp)  .. (1383.0bp,216.0bp) .. controls (1399.4bp,204.27bp) and (1418.0bp,191.44bp)  ..  (13);
  \draw [->] (20) ..controls (1695.3bp,106.18bp) and (1721.5bp,121.4bp)  .. (1733.0bp,144.0bp) .. controls (1798.5bp,272.43bp) and (1748.1bp,327.03bp)  .. (1718.0bp,468.0bp) .. controls (1702.7bp,539.6bp) and (1695.7bp,566.88bp)  .. (1638.0bp,612.0bp) .. controls (1611.6bp,632.63bp) and (1575.9bp,647.57bp)  .. (18);
  \draw [->] (11) ..controls (1023.7bp,209.67bp) and (1049.8bp,193.06bp)  ..  (15);
  \draw [->] (4) ..controls (1083.0bp,999.84bp) and (956.24bp,648.44bp)  .. (25);
  \draw [->] (3) ..controls (524.29bp,873.33bp) and (458.55bp,860.4bp)  .. (413.0bp,828.0bp) .. controls (274.52bp,729.51bp) and (156.5bp,665.29bp)  .. (210.0bp,504.0bp) .. controls (247.69bp,390.35bp) and (269.53bp,351.56bp)  .. (371.0bp,288.0bp) .. controls (476.32bp,222.03bp) and (866.06bp,178.7bp)  .. (14);
  \draw [->] (18) ..controls (1430.8bp,647.26bp) and (1286.2bp,616.76bp)  .. (21);
  \draw [->] (16) ..controls (1698.7bp,133.88bp) and (1703.1bp,120.41bp)  .. (1705.0bp,108.0bp) .. controls (1707.5bp,92.192bp) and (1707.6bp,87.786bp)  .. (1705.0bp,72.0bp) .. controls (1703.5bp,63.139bp) and (1700.8bp,53.768bp)  .. (19);
  \draw [->] (25) ..controls (875.97bp,510.16bp) and (862.4bp,506.48bp)  .. (850.0bp,504.0bp) .. controls (705.9bp,475.23bp) and (531.38bp,459.54bp)  .. (10);
  \draw [->] (6) ..controls (539.15bp,1023.2bp) and (76.0bp,999.77bp)  .. (76.0bp,739.0bp) .. controls (76.0bp,739.0bp) and (76.0bp,739.0bp)  .. (76.0bp,665.0bp) .. controls (76.0bp,502.66bp) and (133.39bp,435.71bp)  .. (277.0bp,360.0bp) .. controls (325.22bp,334.58bp) and (387.12bp,320.03bp)  ..  (8);
  \draw [->] (17) ..controls (1133.3bp,739.04bp) and (1145.5bp,581.19bp)  .. (1147.0bp,576.0bp) .. controls (1208.7bp,365.19bp) and (1204.0bp,265.41bp)  .. (1387.0bp,144.0bp) .. controls (1423.9bp,119.52bp) and (1550.8bp,102.25bp)  .. (20);
  \draw [->] (14) ..controls (1050.4bp,147.54bp) and (1055.8bp,145.41bp)  .. (1061.0bp,144.0bp) .. controls (1166.2bp,115.42bp) and (1501.5bp,97.839bp)  .. (20);
  \draw [->] (24) ..controls (317.26bp,485.61bp) and (318.44bp,456.13bp)  .. (326.0bp,432.0bp) .. controls (349.3bp,357.65bp) and (357.22bp,324.62bp)  .. (426.0bp,288.0bp) .. controls (517.09bp,239.5bp) and (839.06bp,234.94bp)  ..  (11);
  \draw [->] (15) ..controls (1126.7bp,124.85bp) and (1160.9bp,88.689bp)  .. (1199.0bp,72.0bp) .. controls (1280.0bp,36.539bp) and (1550.1bp,23.747bp)  ..  (19);
  \draw [->] (13) ..controls (1519.0bp,124.08bp) and (1612.4bp,65.658bp)  .. (19);
  \draw [->] (2) ..controls (1121.7bp,1240.4bp) and (38.0bp,1231.5bp)  .. (38.0bp,1099.0bp) .. controls (38.0bp,1099.0bp) and (38.0bp,1099.0bp)  .. (38.0bp,449.0bp) .. controls (38.0bp,284.61bp) and (214.09bp,337.02bp)  .. (371.0bp,288.0bp) .. controls (519.14bp,241.72bp) and (912.65bp,225.6bp)  .. (1061.0bp,180.0bp) .. controls (1062.9bp,179.41bp) and (1064.9bp,178.73bp)  ..  (15);
  \draw [->] (23) ..controls (561.51bp,474.15bp) and (536.74bp,409.39bp)  .. (506.0bp,360.0bp) .. controls (499.26bp,349.16bp) and (490.49bp,338.15bp)  ..  (8);
  \draw [->] (3) ..controls (702.18bp,866.1bp) and (989.25bp,828.95bp)  ..  (17);
  \draw [->] (1) ..controls (1225.0bp,1260.2bp) and (1163.4bp,1167.4bp)  ..  (4);
  \draw [->] (22) ..controls (782.17bp,486.85bp) and (748.29bp,453.43bp)  .. (714.0bp,432.0bp) .. controls (683.96bp,413.23bp) and (646.12bp,398.6bp)  ..  (9);
  \draw [->] (9) ..controls (637.94bp,451.97bp) and (753.18bp,663.87bp)  .. (693.0bp,828.0bp) .. controls (678.56bp,867.38bp) and (667.11bp,877.06bp)  .. (632.0bp,900.0bp) .. controls (585.79bp,930.19bp) and (521.78bp,943.7bp)  .. (7);
  \draw [->] (21) ..controls (1193.3bp,652.62bp) and (1205.8bp,757.21bp)  .. (1164.0bp,828.0bp) .. controls (1070.0bp,987.37bp) and (817.6bp,1017.8bp)  ..  (6);
  \draw [->] (12) ..controls (1403.9bp,1071.7bp) and (1240.4bp,718.96bp)  .. (21);
  \draw [->] (18) ..controls (1487.6bp,743.11bp) and (1407.5bp,942.7bp)  .. (1273.0bp,1044.0bp) .. controls (1237.5bp,1070.7bp) and (1187.5bp,1084.6bp)  .. (4);
  \draw [->] (24) ..controls (335.81bp,484.34bp) and (355.81bp,450.54bp)  .. (383.0bp,432.0bp) .. controls (433.51bp,397.57bp) and (505.1bp,385.5bp)  .. (9);
  \draw [->] (18) ..controls (1547.7bp,681.68bp) and (1560.0bp,677.53bp)  .. (1560.0bp,666.0bp) .. controls (1560.0bp,657.62bp) and (1553.5bp,653.14bp)  ..  (18);
  \draw [->] (13) ..controls (1516.7bp,140.21bp) and (1579.6bp,117.15bp)  ..  (20);
  \draw [->] (16) ..controls (1674.9bp,136.28bp) and (1669.3bp,125.71bp)  ..  (20);
  \draw [->] (22) ..controls (794.17bp,485.82bp) and (774.12bp,454.11bp)  .. (751.0bp,432.0bp) .. controls (708.23bp,391.1bp) and (692.76bp,384.73bp)  .. (639.0bp,360.0bp) .. controls (592.05bp,338.4bp) and (534.1bp,323.0bp)  ..  (8);
  \draw [->] (19) ..controls (1719.7bp,33.675bp) and (1732.0bp,29.531bp)  .. (1732.0bp,18.0bp) .. controls (1732.0bp,9.6218bp) and (1725.5bp,5.1433bp)  .. (19);
  \draw [->] (2) ..controls (1372.6bp,1217.9bp) and (1400.2bp,1201.0bp)  .. (12);
  \draw [->] (4) ..controls (1062.7bp,1085.7bp) and (1013.4bp,1071.4bp)  .. (980.0bp,1044.0bp) .. controls (906.75bp,983.98bp) and (922.01bp,940.35bp)  .. (866.0bp,864.0bp) .. controls (773.35bp,737.71bp) and (647.47bp,598.27bp)  .. (23);
  \draw [->] (8) ..controls (411.6bp,322.57bp) and (369.24bp,337.24bp)  .. (338.0bp,360.0bp) .. controls (205.79bp,456.33bp) and (132.0bp,501.42bp)  .. (132.0bp,665.0bp) .. controls (132.0bp,739.0bp) and (132.0bp,739.0bp)  .. (132.0bp,739.0bp) .. controls (132.0bp,872.89bp) and (327.49bp,928.05bp)  .. (7);
  \draw [->] (18) ..controls (1553.4bp,676.05bp) and (1569.0bp,680.59bp)  .. (1583.0bp,684.0bp) .. controls (1661.0bp,703.04bp) and (1753.2bp,720.39bp)  ..  (5);
  \draw [->] (20) ..controls (1578.3bp,91.277bp) and (1484.8bp,98.608bp)  .. (1424.0bp,144.0bp) .. controls (1347.6bp,201.02bp) and (1229.1bp,480.13bp)  .. (21);
  \draw [->] (23) ..controls (591.69bp,566.21bp) and (606.0bp,619.02bp)  .. (606.0bp,665.0bp) .. controls (606.0bp,739.0bp) and (606.0bp,739.0bp)  .. (606.0bp,739.0bp) .. controls (606.0bp,779.36bp) and (598.37bp,825.61bp)  ..  (3);
  \draw [->] (12) ..controls (1365.1bp,1151.3bp) and (1220.9bp,1120.8bp)  ..  (4);
  \draw [->] (9) ..controls (617.69bp,393.68bp) and (630.0bp,389.53bp)  .. (630.0bp,378.0bp) .. controls (630.0bp,369.62bp) and (623.5bp,365.14bp)  .. (9);
  \draw [->] (8) ..controls (456.54bp,348.81bp) and (442.85bp,394.85bp)  .. (10);
  \draw [->] (1) ..controls (1140.9bp,1292.7bp) and (822.6bp,1226.0bp)  .. (649.0bp,1044.0bp) .. controls (613.03bp,1006.3bp) and (597.31bp,945.46bp)  ..  (3);
  \draw [->] (9) ..controls (541.96bp,356.44bp) and (510.4bp,338.7bp)  .. (8);
  \draw [->] (25) ..controls (907.39bp,600.73bp) and (888.56bp,803.97bp)  .. (772.0bp,900.0bp) .. controls (728.28bp,936.02bp) and (558.23bp,947.99bp)  .. (7);
  \draw [->] (17) ..controls (1123.2bp,715.37bp) and (1117.7bp,402.59bp)  .. (1262.0bp,216.0bp) .. controls (1289.5bp,180.45bp) and (1295.1bp,168.1bp)  .. (1333.0bp,144.0bp) .. controls (1438.5bp,76.991bp) and (1584.5bp,40.237bp)  .. (19);
  \draw [->] (12) ..controls (1542.0bp,1159.6bp) and (1706.0bp,1131.2bp)  .. (1706.0bp,1027.0bp) .. controls (1706.0bp,1027.0bp) and (1706.0bp,1027.0bp)  .. (1706.0bp,881.0bp) .. controls (1706.0bp,821.48bp) and (1769.0bp,776.44bp)  .. (5);
  \draw [->] (11) ..controls (999.56bp,317.95bp) and (1039.4bp,563.04bp)  .. (1104.0bp,756.0bp) .. controls (1107.1bp,765.19bp) and (1111.2bp,774.93bp)  .. (17);
  \draw [->] (20) ..controls (1617.0bp,128.42bp) and (1578.6bp,173.17bp)  .. (1553.0bp,216.0bp) .. controls (1508.1bp,290.97bp) and (1511.2bp,316.93bp)  .. (1474.0bp,396.0bp) .. controls (1337.0bp,687.61bp) and (1302.6bp,760.8bp)  .. (1149.0bp,1044.0bp) .. controls (1143.9bp,1053.3bp) and (1138.1bp,1063.4bp)  ..  (4);
  \draw [->] (1) ..controls (1292.7bp,1329.7bp) and (1305.0bp,1325.5bp)  .. (1305.0bp,1314.0bp) .. controls (1305.0bp,1305.6bp) and (1298.5bp,1301.1bp)  ..  (1);
  \draw [->] (7) ..controls (506.16bp,943.19bp) and (568.65bp,929.63bp)  .. (614.0bp,900.0bp) .. controls (649.11bp,877.06bp) and (660.56bp,867.38bp)  .. (675.0bp,828.0bp) .. controls (731.95bp,672.68bp) and (631.8bp,474.56bp)  ..  (9);
  \draw [->] (4) ..controls (943.76bp,1087.1bp) and (246.0bp,1042.8bp)  .. (246.0bp,955.0bp) .. controls (246.0bp,955.0bp) and (246.0bp,955.0bp)  .. (246.0bp,809.0bp) .. controls (246.0bp,710.32bp) and (286.81bp,598.13bp)  ..  (24);
  \draw [->] (23) ..controls (528.48bp,499.2bp) and (481.58bp,478.55bp)  ..  (10);
  \draw [->] (22) ..controls (841.07bp,497.25bp) and (855.65bp,482.78bp)  .. (866.0bp,468.0bp) .. controls (916.85bp,395.37bp) and (908.36bp,365.18bp)  .. (952.0bp,288.0bp) .. controls (957.46bp,278.35bp) and (964.02bp,268.11bp)  ..  (11);
  \draw [->] (9) ..controls (536.19bp,402.65bp) and (485.47bp,424.06bp)  ..  (10);
  \draw [->] (18) ..controls (1560.0bp,646.71bp) and (1594.4bp,632.04bp)  .. (1620.0bp,612.0bp) .. controls (1677.7bp,566.88bp) and (1684.7bp,539.6bp)  .. (1700.0bp,468.0bp) .. controls (1730.1bp,327.03bp) and (1780.5bp,272.43bp)  .. (1715.0bp,144.0bp) .. controls (1707.1bp,128.55bp) and (1692.4bp,116.55bp)  .. (20);
  \draw [->] (5) ..controls (1873.7bp,715.92bp) and (1896.7bp,701.92bp)  .. (1912.0bp,684.0bp) .. controls (2008.3bp,571.55bp) and (2046.0bp,527.06bp)  .. (2046.0bp,379.0bp) .. controls (2046.0bp,379.0bp) and (2046.0bp,379.0bp)  .. (2046.0bp,161.0bp) .. controls (2046.0bp,111.22bp) and (2020.5bp,97.875bp)  .. (1978.0bp,72.0bp) .. controls (1935.3bp,46.035bp) and (1792.3bp,29.531bp)  ..  (19);
  \draw [->] (3) ..controls (669.57bp,935.13bp) and (884.42bp,1064.4bp)  .. (1082.0bp,1116.0bp) .. controls (1200.2bp,1146.9bp) and (1344.9bp,1161.1bp)  ..  (12);
  \draw [->] (11) ..controls (1055.3bp,272.4bp) and (1175.1bp,343.82bp)  .. (1248.0bp,432.0bp) .. controls (1335.2bp,537.53bp) and (1449.0bp,733.18bp)  .. (1449.0bp,881.0bp) .. controls (1449.0bp,1027.0bp) and (1449.0bp,1027.0bp)  .. (1449.0bp,1027.0bp) .. controls (1449.0bp,1067.0bp) and (1449.0bp,1113.3bp)  ..  (12);
  \draw [->] (15) ..controls (1175.1bp,220.35bp) and (1382.4bp,373.18bp)  .. (1550.0bp,504.0bp) .. controls (1649.0bp,581.25bp) and (1764.5bp,676.17bp)  .. (5);
  \draw [->] (24) ..controls (351.76bp,497.6bp) and (374.72bp,481.69bp)  .. (10);
  \draw [->] (7) ..controls (340.06bp,935.92bp) and (114.0bp,882.98bp)  .. (114.0bp,739.0bp) .. controls (114.0bp,739.0bp) and (114.0bp,739.0bp)  .. (114.0bp,665.0bp) .. controls (114.0bp,501.42bp) and (187.79bp,456.33bp)  .. (320.0bp,360.0bp) .. controls (352.04bp,336.66bp) and (395.77bp,321.83bp)  .. (8);
  \draw [->] (15) ..controls (1116.2bp,210.35bp) and (1145.0bp,275.59bp)  .. (1180.0bp,324.0bp) .. controls (1278.8bp,460.51bp) and (1431.6bp,594.96bp)  .. (18);
  \draw [->] (13) ..controls (1494.4bp,223.76bp) and (1560.1bp,353.61bp)  .. (1550.0bp,468.0bp) .. controls (1544.6bp,529.26bp) and (1529.9bp,600.0bp)  .. (18);
  \draw [->] (2) ..controls (1493.3bp,1227.0bp) and (2064.0bp,1172.0bp)  .. (2064.0bp,1099.0bp) .. controls (2064.0bp,1099.0bp) and (2064.0bp,1099.0bp)  .. (2064.0bp,665.0bp) .. controls (2064.0bp,469.87bp) and (1951.1bp,447.68bp)  .. (1839.0bp,288.0bp) .. controls (1815.2bp,254.14bp) and (1811.5bp,242.76bp)  .. (1780.0bp,216.0bp) .. controls (1761.5bp,200.29bp) and (1737.7bp,186.71bp)  .. (16);
  \draw [->] (3) ..controls (768.36bp,834.87bp) and (1612.0bp,609.1bp)  .. (1612.0bp,379.0bp) .. controls (1612.0bp,379.0bp) and (1612.0bp,379.0bp)  .. (1612.0bp,305.0bp) .. controls (1612.0bp,258.12bp) and (1644.6bp,211.74bp)  ..  (16);
  \draw [->] (17) ..controls (1263.3bp,795.66bp) and (1674.9bp,755.08bp)  .. (5);
  \draw [->] (14) ..controls (1115.1bp,240.21bp) and (1434.1bp,508.81bp)  .. (1726.0bp,684.0bp) .. controls (1752.0bp,699.6bp) and (1783.3bp,714.08bp)  .. (5);
  \draw [->] (12) ..controls (1588.2bp,1148.1bp) and (2026.0bp,1074.1bp)  .. (2026.0bp,955.0bp) .. controls (2026.0bp,955.0bp) and (2026.0bp,955.0bp)  .. (2026.0bp,665.0bp) .. controls (2026.0bp,458.33bp) and (2019.9bp,387.78bp)  .. (1905.0bp,216.0bp) .. controls (1880.7bp,179.67bp) and (1875.3bp,166.8bp)  .. (1838.0bp,144.0bp) .. controls (1790.7bp,115.08bp) and (1727.1bp,101.45bp)  .. (20);
  \draw [->] (13) ..controls (1431.8bp,225.77bp) and (1371.6bp,358.98bp)  .. (1318.0bp,468.0bp) .. controls (1294.1bp,516.52bp) and (1281.8bp,525.67bp)  .. (1262.0bp,576.0bp) .. controls (1224.9bp,670.52bp) and (1147.6bp,977.55bp)  .. (4);
  \draw [->] (22) ..controls (791.93bp,588.47bp) and (738.57bp,732.57bp)  .. (658.0bp,828.0bp) .. controls (646.12bp,842.08bp) and (629.81bp,854.58bp)  ..  (3);
  \draw [->] (18) ..controls (1571.0bp,652.28bp) and (1623.4bp,637.9bp)  .. (1662.0bp,612.0bp) .. controls (1846.4bp,488.22bp) and (1990.0bp,457.09bp)  .. (1990.0bp,235.0bp) .. controls (1990.0bp,235.0bp) and (1990.0bp,235.0bp)  .. (1990.0bp,161.0bp) .. controls (1990.0bp,43.814bp) and (1805.4bp,22.889bp)  ..  (19);
  \draw [->] (5) ..controls (1870.7bp,753.68bp) and (1883.0bp,749.53bp)  .. (1883.0bp,738.0bp) .. controls (1883.0bp,729.62bp) and (1876.5bp,725.14bp)  ..  (5);
  \draw [->] (8) ..controls (494.69bp,321.68bp) and (507.0bp,317.53bp)  .. (507.0bp,306.0bp) .. controls (507.0bp,297.62bp) and (500.5bp,293.14bp)  ..  (8);
  \draw [->] (26) ..controls (272.04bp,508.41bp) and (277.17bp,506.07bp)  .. (282.0bp,504.0bp) .. controls (317.07bp,488.96bp) and (357.94bp,473.4bp)  .. (10);
  \draw [->] (19) ..controls (1791.5bp,20.08bp) and (2008.0bp,34.089bp)  .. (2008.0bp,161.0bp) .. controls (2008.0bp,235.0bp) and (2008.0bp,235.0bp)  .. (2008.0bp,235.0bp) .. controls (2008.0bp,457.09bp) and (1864.4bp,488.22bp)  .. (1680.0bp,612.0bp) .. controls (1640.6bp,638.48bp) and (1586.6bp,652.91bp)  ..  (18);
  \draw [->] (2) ..controls (1361.8bp,1181.5bp) and (1411.0bp,1060.8bp)  .. (1411.0bp,955.0bp) .. controls (1411.0bp,955.0bp) and (1411.0bp,955.0bp)  .. (1411.0bp,881.0bp) .. controls (1411.0bp,611.19bp) and (1445.9bp,285.66bp)  .. (13);
  \draw [->] (12) ..controls (1490.1bp,1131.4bp) and (1536.0bp,1080.8bp)  .. (1536.0bp,1027.0bp) .. controls (1536.0bp,1027.0bp) and (1536.0bp,1027.0bp)  .. (1536.0bp,809.0bp) .. controls (1536.0bp,768.56bp) and (1527.6bp,722.33bp)  ..  (18);
  \draw [->] (4) ..controls (1052.8bp,1070.2bp) and (966.0bp,1024.9bp)  .. (966.0bp,955.0bp) .. controls (966.0bp,955.0bp) and (966.0bp,955.0bp)  .. (966.0bp,881.0bp) .. controls (966.0bp,738.34bp) and (940.12bp,697.89bp)  .. (866.0bp,576.0bp) .. controls (858.6bp,563.83bp) and (847.91bp,552.3bp)  ..  (22);
  \draw [->] (23) ..controls (543.71bp,572.63bp) and (492.0bp,657.52bp)  .. (492.0bp,737.0bp) .. controls (492.0bp,811.0bp) and (492.0bp,811.0bp)  .. (492.0bp,811.0bp) .. controls (492.0bp,853.74bp) and (472.57bp,900.22bp)  ..  (7);
  \draw [->] (8) ..controls (505.11bp,327.61bp) and (536.61bp,345.31bp)  ..  (9);
  \draw [->] (15) ..controls (1213.6bp,146.27bp) and (1510.3bp,108.78bp)  .. (20);
  \draw [->] (7) ..controls (533.58bp,953.08bp) and (686.62bp,947.39bp)  .. (803.0bp,900.0bp) .. controls (904.11bp,858.83bp) and (951.74bp,852.91bp)  .. (1002.0bp,756.0bp) .. controls (1034.9bp,692.47bp) and (1004.0bp,666.56bp)  .. (1004.0bp,595.0bp) .. controls (1004.0bp,595.0bp) and (1004.0bp,595.0bp)  .. (1004.0bp,521.0bp) .. controls (1004.0bp,433.25bp) and (995.86bp,411.61bp)  .. (991.0bp,324.0bp) .. controls (989.85bp,303.28bp) and (988.81bp,279.84bp)  ..  (11);
  \draw [->] (8) ..controls (574.96bp,289.94bp) and (850.43bp,253.21bp)  ..  (11);
  \draw [->] (7) ..controls (490.21bp,930.72bp) and (529.32bp,911.43bp)  ..  (3);
  \draw [->] (1) ..controls (1419.4bp,1308.7bp) and (1988.0bp,1273.7bp)  .. (1988.0bp,955.0bp) .. controls (1988.0bp,955.0bp) and (1988.0bp,955.0bp)  .. (1988.0bp,881.0bp) .. controls (1988.0bp,816.52bp) and (1914.0bp,772.57bp)  ..  (5);
  \draw [->] (14) ..controls (1072.8bp,242.53bp) and (1224.0bp,480.88bp)  .. (1408.0bp,612.0bp) .. controls (1431.4bp,628.66bp) and (1460.8bp,642.81bp)  .. (18);
  \draw [->] (17) ..controls (1216.8bp,776.43bp) and (1405.3bp,707.27bp)  ..  (18);
  \draw [->] (4) ..controls (1017.7bp,1080.8bp) and (805.74bp,1046.5bp)  ..  (6);
  \draw [->] (13) ..controls (1484.5bp,188.05bp) and (1498.4bp,202.65bp)  .. (1510.0bp,216.0bp) .. controls (1564.0bp,277.94bp) and (1579.9bp,291.95bp)  .. (1626.0bp,360.0bp) .. controls (1711.5bp,486.3bp) and (1794.8bp,649.53bp)  ..  (5);
  \draw [->] (11) ..controls (1121.5bp,219.57bp) and (1526.4bp,179.14bp)  ..  (16);
  \draw [->] (6) ..controls (717.13bp,987.45bp) and (752.03bp,943.2bp)  .. (772.0bp,900.0bp) .. controls (805.97bp,826.51bp) and (819.57bp,798.56bp)  .. (800.0bp,720.0bp) .. controls (764.96bp,579.36bp) and (737.68bp,545.07bp)  .. (647.0bp,432.0bp) .. controls (636.81bp,419.29bp) and (623.25bp,407.3bp)  ..  (9);
  \draw [->] (23) ..controls (661.17bp,462.84bp) and (883.42bp,307.43bp)  ..  (11);
  \draw [->] (25) ..controls (861.86bp,475.63bp) and (774.17bp,399.69bp)  .. (687.0bp,360.0bp) .. controls (624.31bp,331.46bp) and (544.55bp,317.36bp)  ..  (8);
  \draw [->] (22) ..controls (719.5bp,504.25bp) and (533.19bp,471.24bp)  ..  (10);
  \draw [->] (23) ..controls (580.04bp,479.67bp) and (582.23bp,435.21bp)  ..  (9);
  \draw [->] (19) ..controls (1780.3bp,26.29bp) and (1949.0bp,43.408bp)  .. (1996.0bp,72.0bp) .. controls (2038.5bp,97.875bp) and (2064.0bp,111.22bp)  .. (2064.0bp,161.0bp) .. controls (2064.0bp,379.0bp) and (2064.0bp,379.0bp)  .. (2064.0bp,379.0bp) .. controls (2064.0bp,527.06bp) and (2026.3bp,571.55bp)  .. (1930.0bp,684.0bp) .. controls (1914.3bp,702.3bp) and (1890.8bp,716.52bp)  ..  (5);
  \draw [->] (16) ..controls (1648.2bp,183.21bp) and (1620.1bp,198.92bp)  .. (1598.0bp,216.0bp) .. controls (1436.6bp,340.73bp) and (1434.4bp,415.27bp)  .. (1273.0bp,540.0bp) .. controls (1254.2bp,554.54bp) and (1231.0bp,568.09bp)  ..  (21);
  \draw [->] (5) ..controls (1762.6bp,725.45bp) and (1654.5bp,705.83bp)  .. (1565.0bp,684.0bp) .. controls (1559.1bp,682.56bp) and (1552.9bp,680.92bp)  ..  (18);
  \draw [->] (21) ..controls (1030.6bp,580.66bp) and (500.72bp,537.79bp)  ..  (24);
  \draw [->] (22) ..controls (805.43bp,610.78bp) and (778.01bp,874.07bp)  .. (752.0bp,900.0bp) .. controls (714.58bp,937.3bp) and (555.44bp,948.54bp)  ..  (7);
\end{tikzpicture}}
};
\end{tikzpicture}
\caption{Three implementations of the TBURST4 component of the AMBA bus controller.  Standard synthesis with
Acacia+ produces the state graph on the left with 14 states and 61
cycles. Bounded synthesis produces the graph in the middle with 7
states and 19 cycles. The graph on the right, produced by our tool,
has 7 states and 7 cycles, which is the minimum.}
\label{fig:exampleintro}
\end{figure}

\noindent A first step into this direction is \textit{Bounded Synthesis}~\cite{Schewe:2013}. Here, we bound the number of states of the implementation and can therefore, by incrementally increasing the bound, ensure that the synthesized solution has minimal size.

In this paper, we go one step further by synthesizing implementations
where, additionally, the number of \textit{(simple) cycles} in the state graph is
limited by a given bound. Reducing the number of cycles makes an implementation much
easier to understand.  Compare the three implementations of the TBURST4
component of the AMBA bus controller shown in
Figure~\ref{fig:exampleintro}: standard synthesis with Acacia+
produces the state graph on the left with 14 states and 61
cycles. Bounded Synthesis produces the middle one with 7
states and 19 cycles. The graph on the right, produced by our tool,
has 7 states and 7 cycles, which is the minimum.

An interesting aspect of the number
of cycles as a parameter of the implementations is that the number of cycles that is
potentially needed to satisfy an LTL specification explodes in the
size of the specification: we show that there is a triple exponential lower and
upper bound on the number of cycles that can be enforced by an LTL specification. The impact of the size of the specification on the number of cycles is thus even more dramatic than on the number of states, where the blow-up is double exponential.

Our synthesis algorithm is inspired by Tiernan's cycle counting algorithm from 1970~\cite{Tiernan:1970}.
Tiernan's algorithm is based on exhaustive search. From some arbitrary vertex $v$, the graph is unfolded into a tree such that no vertices repeat on any branch. The number of vertices in the tree that are connected to $v$ then corresponds to the number of cycles through $v$ in the graph. Subsequently, $v$ is removed from the graph, and the algorithm continues with one of the remaining vertices until
the graph becomes empty.
We integrate Tiernan's algorithm into the Bounded Synthesis approach. Bounded Synthesis uses a SAT-solver to simultaneously construct an implementation and a \textit{witness} for the correctness of the implementation~\cite{Schewe:2013}. For the standard synthesis from an LTL specification~$\varphi$, the witness is a finite graph which describes an accepting run of the universal tree automaton corresponding to $\varphi$. To extend the idea to Bounded Cycle Synthesis, we define a second witness that proves the number of cycles, as computed by Tiernan's algorithm, to be equal to or less than the given bound.
An example state graph with three cycles is shown on the left in Figure~\ref{fig:introwitness}. The witness consists of the three graphs shown on the right in Figure~\ref{fig:introwitness}. The first graph proves that vertex 1 is on two cycles (via vertex 2 and vertices 2 and 3). The second graph proves that vertex 2 is on a further cycle, not containing vertex 1, namely via vertex 3. There are no further cycles through vertex 3.

Our experiments show that Bounded Cycle Synthesis is comparable in
performance to standard Bounded Synthesis. The specifications that can
be handled by Bounded Cycle Synthesis are smaller than what can be
handled by tools like Acacia+, but the quality of the synthesized
implementations is much better.  Bounded Cycle Synthesis could be used
in a development process where the programmer decomposes the system
into modules that are small enough so that the implementation can
still be inspected comfortably by the programmer (and synthesized
reasonably fast by using the Bounded Cycle Synthesis
approach). Instead of manually writing the code for such a module, the
programmer has the option of writing a specification, which is then
automatically replaced by the best possible implementation.

\begin{figure}[t]
  \centering

\begin{tikzpicture}
    \clip (-4.2,-2.65) rectangle (6.9,0);

  \node at (-1.5,-1) {
\scalebox{1.2}{\begin{tikzpicture}[>=latex,line join=bevel]
\node (1) at (0bp,0bp) [draw,rounded corners=3,fill=blue!10,inner sep=4pt,minimum width=2em,initial text=,initial left] {1};
\node (2) at (50bp,0bp) [draw,rounded corners=3,fill=blue!10,inner sep=4pt,minimum width=2em] {2};
\node (3) at (100bp,0bp) [draw,rounded corners=3,fill=blue!10,inner sep=4pt,minimum width=2em] {3};

\draw [->] (1) edge[bend left] (2);
\draw [->] (2) edge[bend left] (1);
\draw [->] (2) edge[bend left] (3);
\draw [->] (3) edge[bend left] (2);
\draw [->] (3) edge[bend right=45] (1);
\end{tikzpicture}}
};

\node[anchor=north] at (3.5,0) {
\scalebox{0.9}{\begin{tikzpicture}[>=latex,line join=bevel]
\node (1) at (0bp,0bp) [draw,circle,fill=blue!10,inner sep=3pt,minimum width=2em] {1};
\node (2) at (0bp,-35bp) [draw,circle,fill=blue!10,inner sep=3pt,minimum width=2em] {2};
\node (3) at (20bp,-60bp) [draw,circle,fill=blue!10,inner sep=3pt,minimum width=2em] {3};

\draw [->,blue,dashed] (1) edge[bend left] (2);
\draw [->,red] (2) edge[bend left] (1);
\draw [->,blue,dashed] (2) -- (3);
\draw [->,red] (3) edge[bend right=45] (1);
\end{tikzpicture}}
};

\node[anchor=north] at (4.9,0) {
\scalebox{0.9}{\begin{tikzpicture}[>=latex,line join=bevel]
\node (1) at (0bp,0bp) [draw,circle,fill=blue!10,inner sep=3pt,minimum width=2em] {2};
\node (2) at (0bp,-35bp) [draw,circle,fill=blue!10,inner sep=3pt,minimum width=2em] {3};

\draw [->,blue,dashed] (1) edge[bend left] (2);
\draw [->,red] (2) edge[bend left] (1);
\end{tikzpicture}}
};

\node[anchor=north] at (6.5,0) {
\scalebox{0.9}{\begin{tikzpicture}[>=latex,line join=bevel]
\node (1) at (0bp,0bp) [draw,circle,fill=blue!10,inner sep=3pt,minimum width=2em] {3};
\end{tikzpicture}}
};

\end{tikzpicture}

\caption{Witness for an example state graph with three cycles. The state graph is shown on the left. The first graph on the right proves that vertex 1 is on two cycles (via vertex 2 and vertices 2 and 3). The second graph proves that vertex 2 is on a further cycle not, containing vertex 1, namely via vertex 3. There are no further cycles through vertex 3.}
\label{fig:introwitness}
\end{figure}
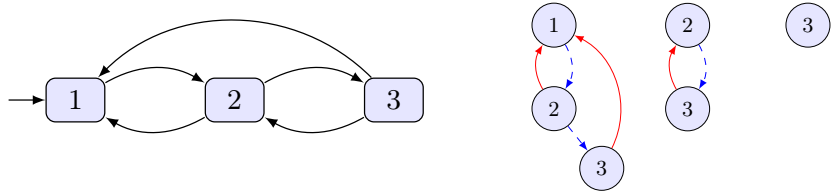

%% file: preliminaries.tex
The non-negative integers are denoted by $ \nats $. An
alphabet~$ \Sigma $ is a non-empty finite set. $ \Sigma^{\omega} $
denotes the set of infinite words over $ \Sigma $. If
$ \alpha \in \Sigma^{\omega} $, then $ \widx{\alpha}{n} $ accesses the
$ n $-th letter of~$ \alpha $, starting at $ \alpha_{0} $. For
the rest of the paper we assume $ \Sigma = 2^{\inputs \cup \outputs} $
to be partitioned into sets of input signals~$ \inputs $ and output
signals~$ \outputs $.

A \textit{Mealy machine}~$ \mealy $ is a tuple
$ (\inputs,\outputs, T, t_{I}, \delta, \lambda) $ over input signals
$ \inputs $ and output signals $ \outputs $, where $ T $ is a finite
set of states, $ t_{I} \in T $ is the initial state,
$ \delta \colon T \times 2^{\inputs} \rightarrow T $ is the transition
function, and $ \lambda \colon T \times 2^{\inputs} \rightarrow 2^{\outputs} $
is the output function. Thereby, the output only depends on the
current state and the last input letter. The size of $ \mealy $,
denoted by $ \size{\mealy} $, is defined as $ \size{T} $. A path~$ p $
of a Mealy machine~$ M $ is an infinite sequence
$ p = (t_{0},\sigma_{0})(t_{1},\sigma_{1})(t_{2},\sigma_{2}) \ldots
\in (T \times \Sigma)^{\omega} $
such that $ t_{0} = t_{I} $,
$ \delta(t_{n},\inputs \cap \sigma_{n}) = t_{n+1} $ and
$ \lambda(t_{n},\inputs \cap \sigma_{n}) = \outputs \cap \sigma_{n} $
for all $ n \in \nats $. We use
$ \pi_{1}(p) = \sigma_{0}\sigma_{1}\sigma_{2}\ldots \in
\Sigma^{\omega} $,
to denote the projection of $ p $ to its second component.
$ \mathcal{P}(\mealy) $ denotes the set of all paths of a Mealy
machine~$ \mealy $.

Specifications are given in \textit{Linear-time Temporal Logic} (LTL). The atomic propositions of the logic consist of the signals $\inputs \cup \outputs$, resulting in the alphabet $\Sigma = 2^{\inputs \cup \outputs}$. The syntax
of an LTL specification~$ \varphi $ is defined as follows:

\medskip

\noindent \quad \ \  $
  \varphi \ \ := \ \ \textit{true} \sep  a \in \inputs \cup \outputs
  \sep \neg \varphi \sep \varphi \vee \varphi
  \sep \LTLnext \varphi \sep \varphi \LTLuntil \varphi
$

\medskip

\noindent The size of a specification~$ \varphi $ is denoted by
$ \size{\varphi} $ and is defined to be the number of sub-formulas of
$ \varphi $. The semantics of LTL are defined over infinite
words~$ \alpha \in \Sigma^{\omega} $. We define the satisfaction
of a word~$ \alpha $ at a position $ n \in \nats $ and a
specification~$ \varphi $, denoted by $ \alpha, n \vDash \varphi $,
for the different choices of $ \varphi $, respectively, as follows:

\begin{itemize}

\item $ \alpha, n \vDash \textit{true} $ 

\item $ \alpha, n \vDash a $ \ iff \ $ a \in \alpha_{i} $

\item $ \alpha, n \vDash \neg \varphi $ \ iff \
  $ \alpha, n \not\vDash \varphi $

\item $ \alpha, n \vDash \varphi_{1} \vee \varphi_{2} $ \ iff \
  $ \alpha, n \vDash \varphi_{1} $ or $ \alpha, i \vDash \varphi_{2} $

\item $ \alpha, n \vDash \LTLnext \varphi $ \ iff \
  $ \alpha, n + 1 \vDash \varphi $

\item $ \alpha, n \vDash \varphi_{1} \LTLuntil \varphi_{2} $ \ iff \
  $ \exists m \geq n.\ \alpha, m \vDash \varphi_{2} $ and
  $ \forall n \leq i < m.\ \alpha, i \vDash \varphi_{1} $

\end{itemize}

\noindent An infinite word~$ \alpha $ satisfies $ \varphi $, denoted
by $ \alpha \vDash \varphi $, iff $ \alpha, 0 \vDash \varphi $.  The
language~$ \lang(\varphi) $ is the set of all words that satisfy
$ \varphi $, i.e.,
$ \lang(\varphi) = \set{ \alpha \in \Sigma^{\omega} \mid \alpha
  \vDash \varphi } $.
Beside the standard operators, we have the standard derivatives of the boolean 
operators, as well as
\mbox{$ \LTLeventually \varphi \ \equiv \ \text{\textit{true}} \LTLuntil
\varphi $}
and
\mbox{$ \LTLglobally \varphi \ \equiv \ \neg \LTLeventually \neg \varphi $}.
A Mealy machine~$ \mealy $ is an implementation of
$ \varphi $ iff
$ \pi_{1}(\mathcal{P}(\mealy)) \subseteq \lang(\varphi) $.










Let $ G = (V,E) $ be a directed graph. A \textit{(simple) cycle}~$ c $
of $ G $ is a a tuple~$ (C,\eta) $, consisting of a non-empty set~$ C
\subseteq V $ and a bijection $ \eta \colon C \mapsto C $ such that 

\begin{itemize}

\item $ \forall v \in C.\ (v,\eta(v)) \in E $ and 

\item $ \forall v \in C.\ n \in \nats.\ \eta^{n}(v) = v
  \ \Leftrightarrow \ n \!\!\mod \size{C} = 0 $,

\end{itemize}

\noindent where $ \eta^{n} $ denotes $ n $ times the application of
$ \eta $.  In other words, a cycle of $ G $ is a path through $ G $
that starts and ends at the same vertex and visits every vertex of
$ V $ at most once. We say that a cycle~$ c = (C,\eta) $ has length
$ n $ iff $ \size{C} = n $.

We extend the notion of a cycle of a graph~$ G $ to Mealy
machines~$ \mealy = (\inputs, \outputs, T, t_{I}, \delta, \lambda) $, such that
$ c $ is a cycle of $ \mealy $ iff $ c $ is a cycle of the graph
$ (T,E) $ for
$ E = \set{ (t,t') \mid \exists \nu \in 2^{\inputs}.\ \delta(t,\nu) =
  t } $.
Thus, we ignore the input labels of the edges of $ \mealy $.  The set
of all cycles of a Mealy machine~$ \mealy $ is denoted by
$ \mathcal{C}(\mealy) $.

A \textit{universal co-Büchi automaton}~$ \aut $ is a
tuple~$ (\Sigma,Q,q_{I},\Delta,R) $, where $ \Sigma $ is the alphabet,
$ Q $ is a finite set of states, $ q_{0} \in Q $ is the initial state,
$ \Delta \subseteq Q \times \Sigma \times Q $ is the transition
relation and $ R \subseteq Q $ is the set of rejecting states. A
run~$ r = (q_{0},\sigma_{0})(q_{1},\sigma_{1})(q_{2},\sigma_{2})\ldots
\in (Q \times \Sigma)^{\omega} $
of $ \aut $ is an infinite sequence such that $ q_{0} = q_{I} $ and
$ (q_{n},\sigma_{n},q_{n+1}) \in \Delta $ for all $ n \in \nats $. A
run~$ r $ is accepting if it has a suffix
$ q_{n}q_{n+1}q_{n+2}\ldots \in (Q \setminus R)^{\omega} $ for some
$ n \in \nats $. An infinite word $ \alpha \in \Sigma^{\omega} $
is accepted by $ \aut $ if all corresponding runs, i.e., all
runs~$ r =
(q_{0},\sigma_{0})(q_{1},\sigma_{1})(q_{2},\sigma_{2})\ldots $
with $ \alpha = \sigma_{0}\sigma_{1}\sigma_{2}\ldots $, are accepting.
The language $ \lang(\aut) $ of $ \aut $ is the set of all
$ \alpha \in \Sigma^{\omega} $, accepted by $ \aut $.

The \textit{run graph~$ G $} of a universal co-Büchi
automaton~$ \aut = (2^{\inputs \cup \outputs},Q,q_{I},\Delta,R) $ and
a Mealy
machine~$ \mealy = (\inputs, \outputs, T, t_{I}, \delta, \lambda) $ is
a directed graph~$ G = (T \times Q,E) $, with
$ E = \set{ ((t,q),(t',q')) \mid \exists \sigma.\
  \delta(t,\inputs \cap \sigma) = t',\, \lambda(t,\inputs \cap \sigma)
  = \mbox{\ensuremath{\outputs \cap \sigma}},\, (q,\sigma,q') \in
  \Delta } $.
A vertex $ (t,q) $ of $ G $ is rejecting iff $ q \in R $. A run graph
is accepting iff there is no cycle of $ G $, which contains a
rejecting vertex. If the run graph is accepting, we say, $ \mealy $ is
accepted by $ \aut $.

%% file: bounds.tex
Our goal is to synthesize systems that have a simple structure. System
quality most certainly has other dimensions as well, but structural
simplicity is a property of interest for most applications.

The purpose of this section is to give theoretical arguments why the
number of cycles is a good measure: we show that the number of cycles
may explode even in cases where the number of states is small, and
even if the specification enforces a large implementation, there may
be a further explosion in the number of cycles. This indicates that
bounding the number of cycles is important, if one wishes to have a
structurally simple implementation. On the other hand, we observe that
bounding the number of states alone is not sufficient in order to
obtain a simple structure.

Similar observations apply to modern programming languages, which tend
to be much better readable than transition systems, because their
control constructs enforce a simple cycle structure. Standard
synthesis techniques construct transition systems, not programs, and
therefore lose this advantage. With our approach, we get closer to the
control structure of a program, without being restricted to a specific
programming language.

\subsection{Upper bounds}

First, we show that the number of cycles of a Mealy
machine~$ \mealy $, implementing an LTL specification~$ \varphi $, is
bounded triply exponential in the size of $ \varphi $. To this end, we
first bound the number of cycles of an arbitrary graph~$ G $ with
bounded outdegree. 

On graphs with arbitrary outdegree, the maximal number of cycles is
given by a fully connected graph, where each cycle describes a
permutation of states, and vice versa. Hence, using standard math we
obtain an upper bound of $ 2^{n \log n} $ cycles for a graph with
$ n $ states. However, our proof uses a more involved argument to
improve the bound even further down to $ 2^{n \log(m + 1)} $ for
graphs with bounded outdegree~$ m $. Such an improvement is desirable,
as for LTL the state graph explodes in the number of states, while the
outdegree is constant in the number of input and output signals.

\begin{lemma}
  Let $ G = (V,E) $ be a directed graph with $ \size{V} = n $ and with
  maximal outdegree $ m $. Then $ G $ has at most
  $ 2^{n \log (m + 1)} $ cycles.
  \label{lem:graph_bound}
\end{lemma}

\begin{proof}
  We show the result by induction over $ n $. The base case is
  trivial, so let $ n > 1 $ and let $ v \in V $ be some arbitrary
  vertex of $ G $.  By induction hypothesis, the subgraph $ G' $,
  obtained from $ G $ by removing $ v $, has at most
  $ 2^{(n - 1) \log (m+1)} $ cycles. Each of these cycles is also a
  cycle in $ G $, thus it remains to consider the cycles of $ G $
  containing $ v $. In each of these remaining cycles, $ v $ has one
  of $ m $ possible successors in $ G' $ and from each such successor
  $ v' $ we have again $ 2^{(n-1) \log (m+1)} $ possible cycles in
  $ G' $ returning to $ v' $. Hence, if we \textit{redirect} these
  cycles to $ v $ instead of $ v' $, i.e., we insert $ v $ before
  $ v' $ in the cycle, then we cover all possible cycles of $ G $
  containing $ v $\footnote{Note that not every such edge needs to exist for
    a concrete given graph. However, in our worst-case analysis,
    every possible cycle is accounted for.}. All
  together, we obtain an upper bound of
  $ 2^{(n - 1) \log (m+1)} + m \cdot 2^{(n - 1) \log (m+1)} = 2^{n
    \log (m+1)} $ cycles in $ G $. \qed
\end{proof}

\noindent We obtain an upper bound on the number of cycles of a Mealy
machine~$ \mealy $.

\begin{lemma}
  Let $ \mealy $ be a Mealy machine. Then
  $ \size{\mathcal{C}(\mealy)} \in \bigo(2^{\size{\mealy} 
    \cdot \size{\inputs}}) $.
  \label{lem:mealy_bound}
\end{lemma}

\begin{proof}
  The Mealy machine~$ \mealy $ has an outdegree of
  $ 2^{\size{\inputs}} $ and, thus, by Lemma~\ref{lem:graph_bound}, the
  number of cycles is bounded by
  $ 2^{\size{\mealy} \log (2^{\size{\inputs}} + 1)} \in
  \bigo(2^{\size{\mealy} \cdot \size{\inputs}}) $. \qed
\end{proof}

\noindent Finally, we are able to derive an upper bound on the
implementations realizing a LTL specification~$ \varphi $.

\begin{theorem}
  Let $ \varphi $ be a realizable LTL specification. Then there is a
  Mealy machine~$ \mealy $, realizing  $ \varphi $, with at most triply
  exponential many cycles in $ \size{\varphi} $.
  \label{thm:upper_bound}
\end{theorem}

\begin{proof}
  From~\cite{Kupferman:2005,Piterman:2006,Schewe:2013} we obtain a
  doubly exponential upper bound in $ \size{\varphi} $ on the size of
  $ \mealy $. With that, applying Lemma~\ref{lem:mealy_bound}
  yields the desired result. \qed
\end{proof}

\subsection{Lower bounds}

It remains to prove that the bound of
Theorem~\ref{thm:upper_bound} is tight. To this end, we show that for
each $ n \in \nats $ there is a realizable LTL
specification~$ \varphi $ with $ \size{\varphi} \in \Theta(n) $, such
that every implementation of $ \varphi $ has at least triply
exponential many cycles in $ n $. The presented proof is inspired
by~\cite{Alur:2004}, where a similar argument is used to prove a
lower bound on the distance of the longest path through a synthesized
implementation~$ \mealy $.  We start with a gadget, which we use to
increase the number of cycles exponentially in the length of the
longest cycle of $ \mealy $.

\begin{lemma}
 \label{lem:blowup_gadget}
  Let $ \varphi $ be a realizable LTL specification, for which every
  implementation~$ \mealy $ has a cycle of length $ n $.  Then there
  is a realizable specification~$ \psi $, such that every Mealy
  machine~$ \mealy' $ implementing $ \psi $ contains at least
  $ 2^{n} $ many cycles.
\end{lemma}
\begin{proof}
  Let $ a $ and $ b $ be a fresh input and output signals,
  respectively, which do not appear in $ \varphi $, and let
  $ \mealy = (\inputs, \outputs, T, t_{I}, \delta, \lambda) $ be an
  arbitrary implementation of~$ \varphi $. We define
  $ \psi ::= \varphi \wedge \LTLglobally(a \leftrightarrow \LTLnext b)
  $ and construct the implementation~$ \mealy' $ as
  \begin{equation*}
    \mealy' = (\inputs \cup \{ a \}, \outputs \cup \{ b
    \}, T \times 2^{\set{ b }}, (t_{I},\emptyset), \delta', \lambda'),
  \end{equation*}
  where $ \lambda'((t,s),\nu) = \lambda(t,\inputs \cap \nu) \cup s $
  and
  \begin{equation*}
    \delta'((t,s), \nu) = 
    \begin{cases}
      (\delta(t,\inputs \cap \nu), \emptyset) & \text{if } a \in \nu \\
      (\delta(t,\inputs \cap \nu), \set{ b }) & \text{if } a \notin \nu
    \end{cases}
  \end{equation*}
  We obtain that $ \mealy' $ is an implementation of $ \psi $.  The
  implementation remembers each input~$ a $ for one time step and then
  outputs the stored value. Thus, it satisfies
  $ \LTLglobally \,(a \leftrightarrow \LTLnext b) $. Furthermore,
  $ \mealy' $ still satisfies $ \varphi $. Hence, $ \psi $ must be
  realizable, too.

  Next, we pick an arbitrary implementation $ \mealy'' $ of $ \psi $,
  which must exist according to our previous observations. Then, after
  projecting away the fresh signals $ a $ and $ b $ from $ \mealy'' $,
  we obtain again an implementation for $ \varphi $, which contains
  a cycle $ (C,\eta) $ of length~$ n $, i.e.,
  $ C = \set{ t_{1}, t_{2}, \ldots, t_{n} } $. We obtain that
  $ \mealy'' $ contains at least the cycles
  \begin{equation*}
    \mathbb{C} = \set{ ( \set{ (t_{i},f(t_{i})) \mid i \in 
        \set{ 1,2,\ldots n }}, (t,s) \mapsto (\eta(t),f(\eta(t)))) 
      \mid f \colon C \rightarrow 2^{\set{ b }} },
  \end{equation*}
  which concludes the proof, since $ \size{\mathbb{C}} = 2^{n} $. \qed
\end{proof}

\noindent Now, with Lemma~\ref{lem:blowup_gadget} at hand, we are
ready to show that the aforementioned lower bounds are tight. The final specification only needs the temporal operators
$ \LTLnext $, $ \LTLglobally $ and $ \LTLeventually $, i.e., the bound
already holds for a restricted fragment~of~LTL.

\begin{theorem}
  For every $ n > 1 $, there is a realizable
  specification~$ \varphi_{n} $ with
  $ \size{\varphi_{n}} \in \Theta(n) $, for which every
  implementation~$ \mealy_{n} $ has at least triply exponential many
  cycles in $ n $.
  \label{thm:lower_bound}
\end{theorem}

\begin{proof}
  According to Lemma~\ref{lem:blowup_gadget}, it suffices to find
  a realizable $ \varphi_{n} $, such that $ \varphi_{n} $ contains at
  least one cycle of length doubly exponential in $ n $. We choose
  \begin{equation*}
    \begin{tikzpicture}
      \draw [decorate,decoration={brace,mirror,amplitude=5pt}] (-2.4,-0.4) --
      (-0.25,-0.4) node [black,midway,yshift=-14] {$ \varphi^{\textit{prem}}_{n} $};
      \draw [decorate,decoration={brace,mirror,amplitude=5pt}] (0.7,-0.4) --
      (2.9,-0.4) node [black,midway,yshift=-14] {$ \varphi^{\textit{con}}_{n} $};

      \node at (0,0) { $ \varphi_{n} \ \ ::= \ \ \LTLglobally \; (\LTLfinally
        \bigwedge\limits_{i=1}^{n} (a_{i} \rightarrow \LTLfinally b_{i})
        \rightarrow \LTLfinally \bigwedge\limits_{i=1}^{n} (c_{i} \rightarrow
        \LTLfinally d_{i})) \ \leftrightarrow \ \LTLglobally \LTLfinally s $};
    \end{tikzpicture}
  \end{equation*}
  \noindent with
  $ \inputs = \inputs_{a} \cup \inputs_{b} \cup \inputs_{c} \cup
  \inputs_{d} $
  and $ \outputs = \set{ s } $, where
  $ \inputs_{x} = \set{ x_{1}, x_{2}, \ldots, x_{n} } $. The
  specification describes a monitor, which checks whether the
  invariant
  $ \LTLfinally \varphi_{n}^{\textit{prem}} \rightarrow \LTLfinally
  \varphi_{n}^{\textit{con}} $
  over the input signals~$ \inputs $ is satisfied or not. Thereby,
  satisfaction is signaled by the output~$ s $, which needs to be
  triggered infinitely often, as long as the invariant stays
  satisfied.

  In the following, we denote a subset $ x \subseteq \inputs_{x} $ by
  the $ n $-ary vector~$ \vec{x} $ over $ \set{ 0, 1 } $, where the
  $ i $-th entry of $ \vec{x} $ is set to $ 1 $ if and only if
  $ x_{i} \in x $.

  The specification~$ \varphi_{n} $ is realizable. First, consider
  that to check the fulfillment of
  \mbox{$ \varphi_{n}^{\textit{prem}} $
    ($ \varphi_{n}^{\textit{con}} $)}, an implementation~$ \mealy $
  needs to store the set of all requests~\mbox{$ \vec{a} $
    ($ \vec{c} $)}, whose $ 1 $-positions have not yet been released
  by a corresponding response~\mbox{$ \vec{b} $
    (\hspace{1pt}$ \vec{\!d} $)}. Furthermore, to monitor the complete
  invariant
  $ \LTLfinally \varphi_{n}^{\textit{prem}} \rightarrow \LTLfinally
  \varphi_{n}^{\textit{con}} $,
  $ \mealy $ has to guess at each point in time, whether
  $ \varphi_{n}^{\textit{prem}} $ will be satisfied in the future
  (under the current request~$ \vec{a} $), or not. To realize this
  guess, $ \mealy $ needs to store a mapping~$ f $, which maps each
  open request~$ \vec{a} $ to the corresponding set of
  requests~$ \vec{c} $~\footnote{Our representation is open for many
    optimizations. However, they will not affect the overall
    complexity result. Thus, we ignore them for the sake of
    readability here.}.  This way, $ \mealy $ can look up the set of
  requests~$ \vec{c} $, tracked since the last occurrence of
  $ \vec{a} $, whenever $ \vec{a} $ gets released by a corresponding
  vector~$ \vec{b} $. If this is the case, it continues to monitor the
  satisfaction of $ \varphi_{n}^{\textit{con}} $ (if not already
  satisfied) and finally adjusts the output signal~$ s $,
  correspondingly. Note that $ \mealy $ still has to continuously
  update and store the mapping~$ f $, since the next satisfaction of
  $ \varphi_{n}^{\textit{prem}} $ may already start while the
  satisfaction of current $ \varphi_{n}^{\textit{con}} $ is still
  checked. There are double exponential many such mappings~$ f $,
  hence, $ \mealy $ needs to be at least doubly exponential in $ n $.

  It remains to show that every such implementation~$ \mealy $
  contains a cycle of at least doubly exponential length. By the
  aforementioned observations, we can assign each state of $ \mealy $
  a mapping~$ f $, that maps vectors~$ \vec{a} $ to sets of
  vectors~$ \vec{c} $. By interpreting the vectors as numbers, encoded
  in binary, we obtain that
  $ f \colon \set{1,2,\ldots,2^{n}} \mapsto 2^{\set{1,2,\ldots,2^{n}}}
  $.
  Next, we again map each such mapping~$ f $ to a binary sequence
  $ b_{f} = b_{0}b_{1}\ldots b_{m} \in \{ 0, 1 \}^{m} $ with
  $ m = 2^{n} $. Thereby, a bit $ b_{i} $ of $ b_{f} $ is set to $ 1 $
  if and only if $ i \in f(i) $. It is easy to observe, that if two
  binary sequences are different, then their related states have to be
  different as well.

  To conclude the proof, we show that the environment has a strategy
  to manipulate the bits of associated sequences~$ b_{f} $ via the
  inputs~$ \inputs $. 

  To set bit~$ b_{i} $, the environment chooses the
  requests~$ \vec{a} $ and $ \vec{c} $ such that they represent $ i $
  in binary. The remaining inputs are fixed to
  $ \vec{b} = \vec{d} = \vec{0} $. Hence, all other bits are not
  affected, as possible requests of previous $ \vec{a} $
  and $ \vec{c} $ remain open.

  To reset bit $ b_{i} $, the environment needs multiple steps. First,
  it picks $ \vec{a} = \vec{c} = \vec{d} = \vec{0} $ and
  $ \vec{b} = \vec{1} $.  This does not affect any bit of the
  sequence~$ b_{f} $, since all requests introduced through
  vectors~$ \vec{c} $ are still open.  Next, the environment executes
  the aforementioned procedure to set bit $ b_{j} $ for every bit
  currently set to $ 1 $, except for the bit $ b_{i} $, it wants to
  reset. This refreshes the requests introduced by previous
  vectors~$ \vec{a} $ for every bit, except for $ b_{i} $.
  Furthermore, it does not affect the sequence $ b_{f} $.  Finally,
  the environment picks $ \vec{a} = \vec{b} = \vec{c} = \vec{0} $ and
  picks $ \vec{d} $ such that it represents $ i $ in binary.  This
  removes $ i $ from every entry in~$ f $, but only resets $ b_{i} $,
  since all other bits are still open due to the previous updates.
  
  With these two operations, the environment can enforce any sequences
  of sequences~$ b_{f} $, including a binary counter counting up to
  $ 2^{2^{n}} $. As different states are induced by the different
  sequences, we obtain a cycle of doubly exponential length in $ n $
  by resetting the counter at every overflow. \qed
\end{proof}

\subsection{The trade-off between states and cycles}

We conclude this section with some observations regarding tradeoffs
between the problem of synthesizing implementations, which are minimal
in the number of states, versus the problem of synthesizing
implementations, which are minimal in the number of cycles. The main
question we answer, is whether we can achieve both: minimality in the
number of states and minimality in the number of
cycles. Unfortunately, this is not possible, as shown by
Theorem~\ref{thm:tradeoffs}.

\begin{theorem}
  \label{thm:tradeoffs}
  For every $ n > 1 $, there is a realizable LTL
  specification~$ \varphi_{n} $ with $ \size{\varphi} \in \Theta(n) $, such that
  \begin{itemize}

  \item there is an implementation of $ \varphi $ consisting of $ n $ states and

  \item there is an implementation of $ \varphi $ containing $ m $ cycles,

  \item but there is no implementation of $ \varphi $ with $ n $
    states and $ m $ cycles.

  \end{itemize}
\end{theorem}

\begin{proof}
  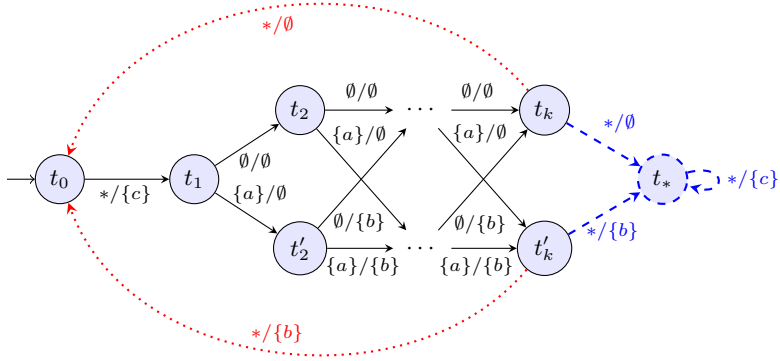
\begin{figure}[t]
    \centering

    \begin{tikzpicture}[node distance=4em]
      \clip (-0.7,-2.37) rectangle (9.6,2.37);

      \node[circle, draw, initial text=, initial left,fill=blue!10] (A) {$ t_{0} $};
      \node[circle, draw,fill=blue!10] (B1) [right of=A, xshift=1.5em] {$ t_{1} $ }; 
      \node[circle, draw,fill=blue!10] (B2) [above right of=B1, xshift=1.5em] {$ t_{2} $ }; 
      \node[circle, draw,fill=blue!10] (C2) [below right of=B1, xshift=1.5em] {$ t_{2}' $}; 
      \node (B3) [right of=B2, xshift=1em, inner sep=0.5em] {$ \cdots $ }; 
      \node (C3) [right of=C2, xshift=1em, inner sep=0.5em] {$ \cdots $}; 
      \node[circle, draw,fill=blue!10] (B4) [right of=B3, xshift=1em] {$ t_{k} $ }; 
      \node[circle, draw,fill=blue!10] (C4) [right of=C3, xshift=1em] {$ t_{k}' $}; 
      \node[circle, draw,dashed, thick, blue,fill=blue!10] (D) [above right of=C4, xshift=2em] {\textcolor{black}{$ t_{*} $}}; 

      \path[->, >=stealth]
      (A) edge node[below right,xshift=-14,yshift=1] {\scriptsize $ * / \set{ c } $ } (B1) 
      (B1) edge node[below,xshift=3] {\scriptsize $ \emptyset / \emptyset $ } (B2) edge node[below,yshift=14,xshift=5] {\scriptsize $ \set{ a } / \emptyset $ } (C2)
      (B2) edge node[above,xshift=1] {\scriptsize $ \emptyset / \emptyset $ } (B3) edge node[above,yshift=10,xshift=-1] {\scriptsize $ \set{ a } / \emptyset $ }(C3)
      (C2) edge node[below,yshift=-10,xshift=-1] {\scriptsize $ \emptyset / \set{ b } $ }(B3) edge node[below, xshift=1] {\scriptsize $ \set{ a } / \set{ b } $} (C3)
      (B3) edge node[above, xshift=-3] {\scriptsize $ \emptyset / \emptyset $ } (B4) edge node[above,yshift=10,xshift=-1] {\scriptsize $ \set{ a } / \emptyset $ }(C4)
      (C3) edge node[below,yshift=-10,xshift=-1] {\scriptsize $ \emptyset / \set{ b } $ } (B4) edge node[below, xshift=-3] {\scriptsize $ \set{ a } / \set{ b } $} (C4)
      ;

      \path[->, >=stealth,red,dotted,thick]
      (C4) edge[in=290,out=230] node[above] {\scriptsize $ * / \set{ b } $ } (A)
      (B4) edge[in=70, out=130,red] node[below] {\scriptsize $ * / \emptyset $ } (A);

      \path[->, >=stealth, dashed, blue, thick]
      (C4) edge node[below, xshift=3] {\scriptsize $ * / \set{ b } $} (D)
      (B4) edge node[above, xshift=6,yshift=1] {\scriptsize $ * / \emptyset $}(D);

      \path[->, >=stealth, thick, blue, dashed]
      (D) edge[loop right] node[right] {\scriptsize $ * / \set{ c } $} (D);

    \end{tikzpicture}

    \vspace{-0.3em}

    \caption{The Mealy automata~$ \mealy_{n} $ (red/dotted) and
      $ \mealy_{n}' $ (blue/dashed). Solid edges are shared between
      both automata.}
    \label{fig:mealies}
  \end{figure}

  Consider the specification
  \vspace{-0.5em}
  \begin{equation*}
    \varphi_{n} = (\neg b \wedge c) \wedge \LTLnext^{k+2} (\neg b \wedge
    c) \wedge \bigwedge_{i = 1}^{k} \LTLnext^{i} (\neg c \wedge \LTLnext \neg c \wedge (a \leftrightarrow
    \LTLnext b))
  \end{equation*} 

  \vspace{-0.3em}

  \noindent over $ \inputs = \set{ a } $ and $ \outputs = \set{ b, c } $, where
  $ \LTLnext^{i} $ denotes $ i $ times the application
  of~$ \LTLnext $. The specification~$ \varphi_{n} $ is realizable
  with at least $ n = 2k + 1 $ states. The corresponding Mealy
  machine~$ \mealy_{n} $ is depicted in
  Figure~\ref{fig:mealies}. However, $ \mealy_{n} $ has $ m = 2^{k} $
  many cycles. This blowup can be avoided by spending the
  implementation at least one more state, which reduces the number of
  cycles to $ m = 1 $. The result~$ \mealy_{n}' $ is also depicted in
  Figure~\ref{fig:mealies}. \qed
\end{proof}

\noindent Our results show that the number of cycles can explode (even more so than the number of states), and that
sometimes this explosion is unavoidable. However, the results also show that
there are cases, where the cycle count can be improved by choosing a
better structured solution.
Hence, it is desirable to have better control over the number of
cycles that appear in an implementation. In the remainder of the
paper, we show how to achieve this control.

%% file: encoding.tex
In this section, we show how to synthesize an implementation~$ \mealy $
from a given LTL specification~$ \varphi $, while giving a guarantee
on the size and the number of cycles of $ \mealy $. We first show how
to guarantee a bound on the number of states of $ \mealy $, by
reviewing the classical Bounded Synthesis approach. Our encoding uses
Mealy machines as implementations, and Boolean Satisfiability (SAT) as
the underlying constraint system. 

We then review the classical algorithm to count the cycles of
$ \mealy $ and show how this algorithm gets embedded into a constraint
system, such that we obtain a guarantee on the number of cycles of
$ \mealy $.

\subsection{Bounded Synthesis}

In the bounded synthesis approach~\cite{Schewe:2013}, we first
translate a given LTL specification~$ \varphi $ into an equivalent
universal co-Büchi automaton~$ \aut $, such that
$ \lang(\aut) = \lang(\varphi) $. Thus, we reduce the problem to
finding an implementation~$ \mealy $ that is accepted by $ \aut $,
i.e., we look for an implementation~$ \mealy $ such that the run graph
of $ \mealy $ and $ \aut $ contains no cycle with a rejecting
vertex. This property is witnessed by a ranking function, which annotates
each vertex of $ G $ by a natural number that bounds the number of
possible visits to rejecting states. The annotation itself is bounded
by $ n \cdot k $, where $ n $ is the size bound on $ \mealy $ and
$ k $ denotes the number or rejecting states of~$ \aut $.

\bigskip

\noindent Fix some set of states $ T $ with $ \size{T} = n $ and let
$ \aut = (2^{\inputs \cup \outputs},Q,q_{I},\Delta,R) $. Then, to
guess a solution within SAT, we introduce the following variables:

\begin{itemize}

\item $ \trans{t}{\nu}{t'} $ for all $ t,t' \in T $ and
  $ \nu \in 2^{\inputs} $, for the transition relation of $ \mealy $.

\item $ \tlabel{t}{\nu}{x} $ for all $ t \in T $,
  $ \nu \in 2^{\inputs} $ and $ x \in \outputs $, for the labels of
  each transition.

\item $ \rgstate{t}{q} $ for all $ t \in T $ and $ q \in Q $, to
  denote the reachable states of the run graph~$ G $ of $ \mealy $ and
  $ \aut $. Only reachable states have to be annotated.

\item $ \annotationd{t}{q}{i} $ for all $ t \in T $, $ q \in Q $ and
  $ 0 < i \leq \log(n \cdot k) $, denoting the annotation of a state
  $ (t,q) $ of $ G $. Thereby, we use a logarithmic number of bits to
  encode the annotated value in binary. We use
  $ \annotation{t}{q} \circ m $ for
  $ \circ \in \set{ <, \leq, =, \geq, > } $, to denote an appropriate
  encoding of the relation of the annotation to some value $ m $ or
  other annotations~$ \annotation{t'}{q'} $.
\end{itemize}

\noindent Given a universal co-Büchi automaton~$ \aut $ and a bound
$ n $ on the states of the resulting implementation, we encode the
Bounded Synthesis problem via the SAT
formula~$ \curlyF_{BS}(\aut,n) $, consisting of the following
constraints:

\begin{itemize}

\item The target of every transition is unambiguous:
  \begin{equation*}
    \bigwedge\limits_{t \in T,\, \nu \in 2^{\inputs}} 
    \textit{exactlyOne}(\set{ 
      \trans{t}{v}{t'} \mid t' \in T})
  \end{equation*}
  where $ \textit{exactelyOne} \colon X \mapsto \mathbb{B}(X) $
  returns a SAT query, which ensures that among all variables of the
  set $ X $ exactly one is \textit{true} and all others are
  \textit{false}.

  \medskip

\item The initial state~$ (t_{I},q_{I}) $ of the run graph for some
  arbitrary, but fix, $ t_{I} \in T $ is reachable and annotated by
  one. Furthermore, all annotations are bounded by
  $ n \cdot k $:
  \begin{equation*}
    \textsc{rgstate}(1,1) \wedge \annotation{1}{1} = 1 \wedge 
    \bigwedge\limits_{t \in T, \, q \in Q}  \annotation{t}{q}
    \leq n \cdot k
  \end{equation*}

\goodbreak

\item Each annotation of a vertex of the run graph bounds the number
  of visited accepting states, not counting the current vertex itself:
  \begin{equation*}
    \bigwedge\limits_{t \in T,\, q \in Q} \rgstate{t}{q} \rightarrow 
    \bigwedge\limits_{\sigma \in 2^{\Sigma}} \textit{label}(t,\sigma)
    \rightarrow \bigwedge\limits_{t' \in T} \trans{t}{\inputs \cap 
      \sigma}{t'} \rightarrow 
  \end{equation*}
  \begin{equation*}
    \qquad \bigwedge\limits_{q' \in \Delta(q,\sigma)} \rgstate{t'}{q'} 
    \wedge \annotation{t}{q} \prec_{q}  \annotation{t'}{q'}
  \end{equation*}
  where $ \prec_{q} $ equals $ < $ if $ q \in R $ and equals $ \leq $
  otherwise. Furthermore, we use the
  function~$ \textit{label}(t,\sigma) $ to fix the labeling of each
  transition, i.e.,
  $ \textit{label}(t,\sigma) = \bigwedge_{x \in \outputs \cap
    \sigma}\, \tlabel{t}{\inputs \cap \sigma}{x} \wedge \bigwedge_{x
    \in \outputs \smallsetminus \sigma} \neg \tlabel{t}{\inputs \cap
    \sigma}{x} $.

  \medskip

\end{itemize}

\begin{theorem}[Bounded Synthesis~\cite{Schewe:2013}]
  For each bound $ n \in \nats $ and each universal co-Büchi
  automaton~$ \aut $, the SAT formula~$ \curlyF_{BS}(\aut,n) $ is
  satisfiable if and only if there is a Mealy machine~$ \mealy $ with
  $ \size{\mealy} = n $, which is accepted by $ \aut $.
\end{theorem}

\subsection{Counting Cycles}

Before we bound the number of cycles of a Mealy machine~$ \mealy $, we
review Tiernan's classical algorithm \cite{Tiernan:1970} to count the
number of cycles of a directed graph~$ G $. The algorithm not only
gives insights into the complexity of the problem, but also contains
many inspirations for our latter approach.

\goodbreak

\smallskip

\paragraph{Algorithm 1.} Given a directed graph~$ G = (V,E) $, we
count the cycles of $ G $ using the following algorithm:

\begin{enumerate}

\item[(1)] Initialize the cycle counter~$ c $ to $ c := 0 $ and some
  set~$ P $ to $ P := \emptyset $.

\item[(2)] Pick some arbitrary vertex~$ v_{r} $ of $ G $, set
  $ v := v_{r} $ and $ P := \set{ v_{r} } $.

\item[(3)] For all edges $ (v,v') \in E $, with
  $ v' \notin P \setminus \set{ v_{r} } $:

\begin{enumerate}

\item[(3a)] If $ v' = v_{r} $, increase $ c $ by one.

\item[(3b)] Oherwise, add $ v' $ to $ P $ and recursively execute
  (3). Afterwards, reset $ P $ to its value before the recursive call.

\end{enumerate}

\item[(4)] Obtain the sub-graph~$ G' $, by removing $ v_{r} $ from $ G $:

\begin{enumerate}

\item[(4a)] If $ G' $ is empty, return $ c $.

\item[(4b)] Otherwise, continue from (2) with $ G' $.

\end{enumerate}

\end{enumerate}

\noindent The algorithm starts by counting all cycles that contain
the first picked vertex~$ v_{r} $. This is done by an unfolding of the
graph into a tree, rooted in $ v_{r} $, such that there is no repetition
of a vertex on any path from the root to a leaf. The number of
vertices that are connected to the root by an edge of $ E $ then
represents the corresponding number of cycles through $ v_{r} $. The
remaining cycles of $ G $ do not contain $ v_{r} $ and, thus, are
cycles of the sub-graph $ G' $ without $ v_{r} $, as well. Hence, we
count the remaining cycles by recursively counting the cycles of
$ G' $. The algorithm terminates as soon as $ G' $ gets empty.

The algorithm is correct~\cite{Tiernan:1970}, but has the drawback,
that the unfolded trees, may become exponential in the size of the graph,
even if none of their vertices is connected to the root, i.e., even if there
is no cycle to be counted. For an example consider the induced graph
of~$ \mealy_{n}' $, as depicted in Figure~\ref{fig:mealies}. However,
this drawback can be avoided by first reducing the graph to all its
strongly connected components (SCCs) and then by counting the cycles
of each SCC separately~\cite{Weinblatt:1972,Johnson:1975}. 
A cycle never leaves an SCC of the graph.

As a result, we obtain an improved algorithm, which is exponential in
the size of $ G $, but linear in the number of cycles~$ m $.
Furthermore, the time between two detections of a cycle, during the
execution, is bounded linear in the size of the graph~$ G $.

\subsection{Bounded Cycle Synthesis}

We combine the insights of the previous sections to obtain
a synthesis algorithm, which not only bounds the number of states of
the resulting implementation~$ \mealy $ but also bounds the number of
cycles of~$ \mealy $. We use the unfolded trees from the
previous section as our witnesses.

We call a tree that witnesses $ m $ cycles in $ G $, all containing
the root~$ r $ of the tree, a witness-tree~$ \tree_{r,m} $ of $ G $.
Formally, a \emph{witness-tree}~$ \tree_{r,m} $ of $ G = (V,E) $ is a labeled
graph $ \tree_{r,m} = ((W,B\cup R),\tau) $, consisting of a graph
$ (W,B \cup R) $ with $ m = \size{R} $ and a labeling function
$ \tau \colon W \rightarrow V $, such that:

\begin{enumerate}

\item The edges are partitioned into blue edges~$ B $ and red
  edges~$ R $.

  \smallskip

  \label{con:witness_tree_first}

\item All red edges lead back to the root:

  \smallskip

  \quad $ R \subseteq W \times \set{ r } $

  \smallskip

  \label{con:witness_tree_red}  
  
\item No blue edges lead back to the root:
  
  \smallskip

  \quad $ B \cap W \times \set{ r } = \emptyset $
  
  \smallskip

  \label{con:witness_noblue}  

\item Each non-root has at least one blue incoming edge:

  \smallskip

  \quad $ \forall w' \in W \setminus \set{ r }.\ \exists w \in W.\ (w,w')
  \in B $ 

  \smallskip

\item Each vertex has at most one blue incoming edge:

  \smallskip

  \quad $ \forall w_{1},w_{2},w \in W.\ (w_{1},w) \in B \wedge (w_{2},w) \in
  B \Rightarrow w_{1} = w_{2} $

  \smallskip

  \label{con:witness_mostblue}  

\item The graph is labeled by an unfolding of $ G $:

  \smallskip
  
  \quad $ \forall w,w' \in B \cup R.\ (\tau(w),\tau(w')) \in E $, 

  \smallskip

\item The unfolding is complete: 

  \smallskip

  \quad
  $ \forall w \in W.\ \forall v' \in V.\ (\tau(w),v') \in E
  \Rightarrow \exists w' \in W.\ (w,w') \in B \cup R \wedge \tau(w') =
  v' $

  \smallskip

  \label{con:witness_tree_completness}

\item Let $ w_{i}, w_{j} \in W $ be two different vertices that appear
  on a path from the root to a leaf in the $ r $-rooted tree
  $ (W,B) $\footnote{Note that the tree property is enforced by
    Conditions~\ref{con:witness_noblue} --
    \ref{con:witness_mostblue}.}. Then the labeling of $ w_{i} $ and
  $ w_{j} $ differs, i.e., $ \tau(v_{i}) \neq \tau(v_{j}) $.

  \smallskip

  \label{con:witness_tree_nodouble}

\item The root of the tree is the same as the corresponding vertex of
  $ G $, i.e., $ \tau(r) = r $.

  \label{con:witness_tree_last}

\end{enumerate}

\begin{lemma}
  Let $ G = (V,E) $ be a graph consisting of a single SCC, $ r \in V $
  be some vertex of $ G $ and $ m $ be the number of cycles of $ G $
  containing $ r $. Then there is a
  witness-tree~$ \tree_{r,m} = ((W,B \cup R),\tau) $ of $ G $ with
  $ \size{W} \leq m \cdot \size{V} $.
  \label{lem:cycles_to_witness}
\end{lemma}

\begin{proof}
  We construct $ \tree_{r,m} $ according to the strategy induced by
  \textit{Algorithm~1}, where an edge is colored red if and only if it
  leads back to the root. The constructed tree satisfies all
  conditions \ref{con:witness_tree_first} --
  \ref{con:witness_tree_last}. By correctness of \textit{Algorithm~1},
  we have that $ \size{R} = m $.

  Now, for the sake of contradiciton, assume
  $ \size{W} > m \cdot \size{V} $. First we observe, that the depth of
  the tree~$ (W,B) $ must be bounded by $ \size{V} $ to satisfy
  Condition~\ref{con:witness_tree_nodouble}. Hence, as there are at
  most $ m $ red edges in $ \tree_{r,m} $, there must be a
  vertex~$ w \in W $ without any outgoing edges. However, since $ G $
  is a single SCC, this contradicts the completeness of $ \tree_{r,m} $
  (Condition~\ref{con:witness_tree_completness}). \qed
\end{proof}

\begin{lemma}
  Let $ G = (V,E) $ be a graph consisting of a single SCC and let
  $ \tree_{r,m} $ be a witness-tree of $ G $. Then there are at most
  $ m $ cycles in $ G $ that contain $ r $.
  \label{lem:witness_to_cycles}
\end{lemma}

\begin{proof}
  Let $ \tree_{r,m} = ((W,R \cup B),\tau) $. Assume for the sake of
  contradiction that $ G $ has more than $ m $ cycles and let
  $ c = (C,\eta) $ be an arbitrary such cycle. By the completeness of
  $ \tree_{r,m} $, there is path $ w_{0}w_{1}\ldots w_{\size{C}-1} $
  with $ w_{0} = r $ and $ \tau(w_{i}) = \eta^{i}(r) $ for all
  $ 0 \leq i < \size{C} $. From $ w_{i} \neq r $ and
  Condition~\ref{con:witness_tree_red}, it follows
  $ (w_{i-1},w_{i}) \in B $ for all $ 0 < i < \size{C} $. Further,
  $ \eta^{\size{C}}(r) = r $ and thus
  $ (w_{\size{C}-1},w_{0}) \in R $. Hence, by the tree shape of
  $ (W,B) $, we get $ \size{R} > m $, yielding the desired
  contradiction. \qed
\end{proof}

\noindent From Lemma~\ref{lem:cycles_to_witness}
and~\ref{lem:witness_to_cycles} we derive that $ \tree_{r,m} $ is a
suitable witness to bound the number of cycles of an
implementation~$ \mealy $. Furthermore, from
Lemma~\ref{lem:cycles_to_witness} we also obtain an upper bound on the
size of $ \tree_{r,m} $.

\bigskip

\noindent We proceed with our final encoding. Therefore, we first
construct a simple directed graph~$ G $ out of the
implementation~$ \mealy $.  Then, we guess all the sub-graphs,
obtained from $ G $ via iteratively removing vertices, and split them
into their corresponding SCCs. Finally, we guess the witness-tree for
each such SCC.

To keep the final SAT encoding compact, we even introduce some further
optimizations. First, we do not need to introduce a fresh copy for
each SCC, since the SCC of a vertex is always unique. Thus, it
suffices to guess an annotation for each vertex, being unique for each
SCC. Second, we have to guess $ n $ trees~$ \tree_{i,r_{i}} $, each
one consisting of at most $ i \cdot n $ vertices, such that the sum of
all $ i $ is equal to the overall number of cycles~$ m $. One possible
solution would be to overestimate each $ i $ by $ m $. Another
possibility would be to guess the exact distribution of the cycles
over the different witness-trees~$ \tree_{i,r_{i}} $. However, there
is a smarter solution: we guess all trees together in a single graph
bounded by $ m \cdot n $. Additionally, to avoid possible
interleavings, we add an annotation of each vertex by its
corresponding witness-tree~$ \tree_{i,r_{i}} $. Hence, instead of
bounding the number of each $ \tree_{i,r_{i}} $ separately by~$ i $,
we just bound the number of all red edges in the whole forest by~$ m
$.
This way, we not only reduce the size of the encoding, but also skip
the additional constrains, which would be necessary to sum the
different witness-tree bounds~$ i $ to $ m $, otherwise.

\medskip

\noindent Let $ T $ be some ordered set with $ \size{T} = n $ and
$ S = T \times \set{ 1,2,\ldots, m } $. We use $ T $ to denote the
vertices of $ G $ and $ S $ to denote the vertices of the forest of
$ \tree_{i,r_{i}} $\,s. Further, we use $ M = T \times \set{ 1 } $ to
denote the roots and $ N = S \setminus M $ to denote the non-roots of
the corresponding trees.  We introduce the following variables:

\begin{itemize}

\item $ \edge{t}{t'} $ for all $ t,t' \in T $, denoting the edges of
  the abstraction of $ \mealy $ to $ G $.

\item $ \bedge{s}{s'} $ for all $ s \in S $ and
  $ s' \in N $, denoting a blue edge.

\item $ \redge{s}{s'} $ for all $ s \in S $ and $ s' \in M $, denoting a
  red edge.

\item $ \wtreed{s}{i} $ for all $ s \in S $, $ 0 < i \leq \log n $,
  denoting the witness-tree of each $ s $. As
  before, we use $ \wtree{s} \circ x $ to relate values with the
  underlying encoding.

\item $ \allowed{s}{t} $ for all $ s \in S $ and $ t \in T $, denoting
  the set of all vertices $ t $, already visited at $ s $, since
  leaving the root of the corresponding witness-tree.

\item $ \rboundd{c}{i} $ for all $ 0 < c \leq m $,
  $ 0 < i \leq \log (n \cdot m) $, denoting an ordered list of all red
  edges, bounding the red edges of the forest.

\item $ \sccd{k}{t}{i} $ for all $ 0 < k \leq n $, $ t \in T, $ and
  $ 0 \leq i < \log n $, denoting the SCC of $ t $ in the $ k $-th
  sub-graph of $ G $. The sub-graphs are obtained 
  by iteratively removing vertices of $ T $, according to the
  pre-defined order. This way, each sub-graph contains exactly all
  vertices that are larger than the root.
 
\end{itemize}

\noindent Note that by the definition of $ S $ we introduce $ m $
explicit copies for each vertex of~$ G $. This is sufficient, since
each cycle contains each vertex at most once. Thus, the labeling
$ \tau $ of a vertex $ s $ can be directly derived from the first
component of~$ s $.

 Given a universal co-Büchi automaton~$ \aut $, a bound $ n $ on the
states of the resulting implementation~$ \mealy $, and a bound $ m $
on the number of cycles of $ \mealy $, we encode the Bounded Cycle
Synthesis problem via the SAT formula
$ \curlyF_{BS}(\aut,n) \wedge \curlyF_{CS}(\aut,n,m) \wedge
\curlyF_{SCC}(n) $.
The constraints of $ \curlyF_{CS}(\aut,n,m) $, bounding the cycles
of the system, are given by Table~\ref{tab:constraints}. The
constraints of $ \curlyF_{SCC}(n) $, enforcing that each
vertex is labeled by a unique SCC, is given in
Appendix~\ref{apx:encoding}.

\begin{theorem}
  For each pair of bounds $ n, m \in \nats $ and each universal
  co-Büchi automaton~$ \aut $ with $ \size{\aut} = k $, the
  formula~$ \curlyF = \curlyF_{BS}(\aut,n) \wedge
  \curlyF_{CS}(\aut,n,m) \wedge \curlyF_{SCC} $
  is satisfiable if and only if there is a Mealy machine~$ \mealy $
  with $ \size{\mealy} = n $ and
  \mbox{$ \size{\mathcal{C}(\mealy)} = m $}, accepted by $ \aut $.
  Furthermore, $ \curlyF $ consists of $ x $ variables with
  \mbox{$ x \in \bigo\hspace{0.5pt}(n^{3} \!+ n^{2}(m^{2}\!+
    2^{\size{\inputs}}) \hspace{-1.5pt}+\hspace{-1pt} n
    \size{\outputs}\hspace{-1.5pt}+\hspace{-1pt} nk \, \log(nk)) $}
  and
  \mbox{$ \size{\curlyF} \in \bigo\hspace{0.5pt}(n^{3}
    \!+\hspace{-0.5pt} n^{2}(m^{2} \!+ k \size{\Sigma})) $}.
\end{theorem}

\begin{table}[H]
  \centering
\caption{Constraints of the SAT formula~$ \curlyF_{CS}(\aut, n, m) $.}
\label{tab:constraints}
\renewcommand{\arraystretch}{1.3}
\begin{tabular}{|>{\centering}m{0.17\textwidth} m{0.43\textwidth} | >{\small}m{0.4\textwidth-12pt} |}
  \hline
  $ \bigwedge\limits_{t,t' \in T, \nu \in 2^{I}} $ & $ \trans{t}{\nu}{t'} \rightarrow \edge{t}{t'} $ & 
  {\multirow{2}{0.36\textwidth}{\centering \raisebox{-6pt}{Construction of $ G $ from $ \mealy $.}}}
  \\ \cline{1-2}
  $ \bigwedge\limits_{t,t' \in T} $ & $ \edge{t}{t'} \rightarrow \bigvee\limits_{\nu \in 2^{I}} \trans{t}{\nu}{t'} $ & 
  \\ \hline
  $ \bigwedge\limits_{r \in T} $ & $ \wtree{(r,1)} = r $
  & Roots indicate the witness-tree.
  \\ \hline
  $ \bigwedge\limits_{s \in S,\, (r,1) \in M} $ & $\redge{s}{(r,1)} 
  \rightarrow \wtree{s} = r $
  & Red edges only connect vertices of the current 
    $ \tree_{i,r_{i}} $. 
  \\ \hline
  $ \bigwedge\limits_{s \in S,\, s' \in N} $ & $ \bedge{s}{s'} \newline
                                               \mbox{\quad} \rightarrow \wtree{s} = \wtree{s'} $ 
  & Blue edges only connect vertices of the
    current $ \tree_{i,r_{i}} $.
  \\ \hline
  $ \bigwedge\limits_{s' \in N} $ & $ \textit{exactlyOne}(\newline \mbox{\quad} \set{ \bedge{s}{s'} \mid s \in S }\ ) $
  & Every non-root has exactly one blue incoming edge.
  \\ \hline
  $ \bigwedge\limits_{(t,c) \in S,\, r \in T,} $ & $ \redge{(t,c)}{(r,1)} \rightarrow \edge{t}{r} $
  & Red edges are related to the edges of the graph~$ G $.
  \\ \hline
  $ \bigwedge\limits_{(t,c) \in S,\, (t',c') \in N} $ & $ \bedge{(t,c)}{(t',c')} \rightarrow \edge{t}{t'} $  
  & Blue edges are related to the edges of the graph~$ G $.
  \\ \hline
  $ \bigwedge\limits_{\substack{(t,c) \in S,\, r \in T,\\ t \geq r}}  $ &
  $ \edge{t}{r} \wedge \scc{r}{t} = \scc{r}{r} \wedge \newline \wtree{(t,c)} = r \newline \mbox{\quad} \rightarrow \redge{(t,c)}{(r,1)} $ & Every possible red edge must be taken.
  \\ \hline
  $ \bigwedge\limits_{\substack{(t,c) \in S,\, r,t' \in T,\\ t \geq t'}}  $ &
  $ \edge{t}{t'} \wedge \scc{r}{t} = \scc{r}{t'} \wedge \newline \wtree{(t,c)} = r \wedge \allowed{(t,c)}{t'} \newline \mbox{\quad} \rightarrow \bigvee\limits_{0 < c' \leq m} \bedge{(t,c)}{(t',c')} $ & Every possible blue edge must be taken.
  \\ \hline
  $ \bigwedge\limits_{r \in T} $ & $ \bigwedge\limits_{t \leq r} \neg \allowed{(r,1)}{t} \wedge \newline \bigwedge\limits_{t > r} \allowed{(r,1)}{t} $ & 
  Only non-roots of the corresponding sub-graph can be successors of a root.
  \\ \hline
  $ \bigwedge\limits_{(t,c) \in S,\, s \in N} $ & $ \bedge{(t,c)}{s} \newline \mbox{\ \ } \rightarrow \neg \allowed{s}{t} \wedge \newline \mbox{\quad \quad\,} (\allowed{s}{t'} \newline \mbox{\qquad\quad} \leftrightarrow \allowed{(t,c)}{t'}) $ & Every vertex appears at most once on a path from the root to a leaf. 
  \\ \hline

  $ \bigwedge\limits_{s \in S,\, s' \in M} $ & $ \redge{s}{s'} \newline \mbox{\ \ } \rightarrow \bigvee\limits_{0 < c \leq m} \rbound{c} = f(s) $ & The list of red edges is complete. ($ f(s) $ maps each state of $ S $ to a unique number in $ \set{ 1,\ldots,n\cdot m } $)
  \\ \hline
  $ \bigwedge\limits_{0 < c \leq m} $ & $ \rbound{c} < \rbound{c + 1} $ & Red edges are strictly ordered.
  \\ \hline
\end{tabular}

\renewcommand{\arraystretch}{1}
\end{table}

%% file: results.tex
We have implemented the Bounded Cycle Synthesis approach in our tool
\textit{Bo\!WSer}, the Bounded Witness Synthesizer, and compared it
against standard Boun\-ded Synthesis and
\textit{Acacia+}~(v2.3)~\cite{Filiot:2011,Filiot:2013}. To ensure a
common encoding, we used \textit{Bo\!WSer} for both, the Bounded
Synthesis and the Bounded Cycle Synthesis approach. Our tool uses
\textit{LTL3BA}~(v1.0.2)~\cite{Babiak:2012} to convert specifications
to universal co-Büchi automata. The created SAT queries are solved by
\textit{MiniSat}~(v.2.2.0)~\cite{Een:2003} and
\textit{clasp}~(v.3.1.4)~\cite{Gebser:2007}, where the result of the
faster solver is taken.

The benchmarks are given in TLSF~\cite{Jacobs:2016} and represent a
decomposition of \textsc{Arm}'s \textit{Advanced Microcontroller Bus
  Architecture} (AMBA)~\cite{Amba:1999}. They are created from the
assumptions and guarantees presented in~\cite{Jobstmann:2007}, which
were split into modules, connected by new signals. A detailed
description of the benchmarks is given in~\cite{Jacobs:2016}.

All experiments were executed on a Unix machine, operated by a 64-bit
kernel (v4.1.12) running on an Intel Core i7 with 2.8GHz and 8GB
RAM. Each experiment had a time limit of 1000 seconds and a memory
limit of 8GB. When counting cycles of a solution, the limit was set to
10000000 cycles.

\newbox{\autUCTb}
\savebox{\autUCTb}{\ensuremath{\aut_{\scalebox{0.9}{\ensuremath{UCT}}}}}
\newcommand{\autUCT}{\ensuremath{\usebox{autUCTb}}}

\begin{table}[t]
  \centering
  \caption{Results of the tools \textit{LTL3BA}, \textit{Aca(cia)+} and \textit{BoWSer}. 
    The \textit{LTL3BA} tool was used to generate the universal co-Büchi 
    tree automata $ \aut_{\protect\scalebox{0.5}{\ensuremath{UCT}}} $. The Bo(unded) Sy(nthesis) and 
    Bo(unded) Cy(cle Synthesis) encodings were generated 
    by BoWSer.}
  \label{tab:results}
  \renewcommand{\arraystretch}{1.3}
  \scalebox{0.857}{
    \begin{tabular}{|l||c||c|c||c|c|c||c|c|c||c|c|}
      \hline
      \multirow{3}{*}{\textbf{Benchmark}}
      & \multicolumn{3}{c||}{\textbf{Size}} 
      & \multicolumn{3}{c||}{\textbf{Cycles}}
      & \multicolumn{5}{c|}{\textbf{Time (s)}} \\
      \cline{2-12}
      & \multirow{2}{*}{$ \aut_{\scalebox{0.5}{\ensuremath{UCT}}} $} & \multirow{2}{*}{Aca+} & BoSy/
      & \multirow{2}{*}{Aca+} & \multirow{2}{*}{BoSy}  & \multirow{2}{*}{BoCy}
      & \multirow{2}{*}{\,Aca+\,} & \multicolumn{2}{c|}{SAT} & \multicolumn{2}{c|}{UNSAT} \\
      \cline{9-12}
      & & & BoCy & & & & & BoSy      & BoCy       & BoSy      & BoCy   \\

      \cline{1-7} \cline{9-12}

      \hline
      \hline
      ARBITER[2]         & 6 & 26   & \cb 2 & 5439901       & \cb 3  & \cb 3  & \cb 0.261  & 0.847     & 0.868     & 0.300  & 0.836   \\
      \hline
      ARBITER[3]         & 20 & 111  & \cb 3 & \textcolor{red!70!black}{$ > 9999999 $} & 8      & \cb 4  & \cb 0.511  & 9.170     & 9.601     & 3.916  & 9.481   \\
      \hline
      ARBITER[4]         & 64 & 470  & \cb 4 & \textcolor{red!70!black}{$ > 9999999 $} & 8      & \cb 5  & \cb 12.981 & 105.527   & 109.180   & 56.853 & 106.803 \\
      \hline
      LOCK[2]            & 12 & 4    & \cb 3 & 12            & 6      & \cb 5  & 0.459      & \cb 0.395 & 0.522     & 0.165  & 0.487   \\ 
      \hline
      LOCK[3]            & 20 & 4    & \cb 3 & 12            & \cb 5  & \cb 5  & 55.917     & \cb 1.037 & 1.245     & 0.433  & 1.107   \\
      \hline
      LOCK[4]            & 36 & --  & \cb 3 & --  & 6      & \cb 5  & \textcolor{red!70!black}{$ > 999 $}  & \cb 4.419 & 4.761 & 1.407  & 3.726   \\    
      \hline
      ENCODE[2]          & 3 & 6    & \cb 2 & 41            & \cb 3  & \cb 3  & 0.473      & \cb 0.071 & 0.089      & 0.048  & 0.084   \\
      \hline
      ENCODE[3]          & 5 & 16   & \cb 3 & 90428         & \cb 8  & \cb 8  & 1.871      & \cb 0.292 & 0.561      & 0.200  & 0.503   \\
      \hline
      ENCODE[4]          & 5 & 20   & \cb 4 & \textcolor{red!70!black}{$ > 9999999 $} & \cb 24 & \cb 24 & 4.780      & \cb 1.007 & 16.166 & 0.579  & \textcolor{red!70!black}{$ > 999 $}  \\ 
      \hline
      DECODE             & 1 & 4    & \cb 1 & 8             & \cb 1  & \cb 1  & 0.328      & 0.055     & \cb 0.051  & --      & --       \\
      \hline
      SHIFT              & 3 & 6    & \cb 2 & 31            & \cb 3  & \cb 3  & 0.387      & \cb 0.060 & 0.072      & 0.041  & 0.071   \\
      \hline 
      TBURST4            & 103 & 14   & \cb 7 & 61            & 19     & \cb 7  & \cb 0.634  & 8.294     & 206.604  & 6.261  & \textcolor{red!70!black}{$ > 999 $} \\
      \hline
      TINCR              & 43 & 5    & \cb 3 & 7             & 5      & \cb 2  & \cb 0.396  & 2.262     & 2.279        & 0.845  & 2.221   \\
      \hline
      TSINGLE            & 22 & 8    & \cb 4 & 12            & 5      & \cb 4  & \cb 0.372  & 1.863     & 2.143      & 1.165  & 2.067   \\
      \hline
    \end{tabular}
  }
  \renewcommand{\arraystretch}{1}
\end{table}

\medskip

\noindent The results of the evaluation are shown in
Table~\ref{tab:results}, which displays the sizes of the intermediate
universal co-Büchi tree
automata~$ \aut_{\protect\scalebox{0.5}{\ensuremath{UCT}}} $, the
sizes of the synthesized implementations~$ \mealy $, the number of
cycles of each implementation~$ \mealy $, and the overall synthesis
time. Thereby, for each instance, we guessed the minimal number of
states for the Bounded Synthesis approach and, additionally, the
minimal number of cycles for the Bounded Cycle Synthesis approach, to
obtain a single satisfiable instance. Further, to verify the result,
we also created the unsatisfiable instance, where the state bound was
decreased by one in the case of Bounded Synthesis and the cycle bound
was decreased by one in the case of Bounded Cycle Synthesis. Note that
these two instances already give an almost complete picture, since for
increased and decreased bounds the synthesis times behave
monotonically. Hence, increasing the bound beyond the first realizable
instance increases the synthesis time. Decreasing it below the last
unsatisfiable instance decreases the synthesis time.  The results for
the TBURST4 component are additionally depicted in
Figure~\ref{fig:exampleintro}.

On most benchmarks, Acacia+ solves the synthesis problem the fastest,
followed by Bounded Synthesis and our approach. (On some benchmarks,
Bounded Synthesis outperforms Acacia+.)  Comparing the running times
of Bounded Synthesis and Bounded Cycle Synthesis, the overhead for
bounding the number of cycles is insignificant on most benchmarks.
The two exceptions are ENCODE, which requires a fully connected
implementation, and TBURST4, where the reduction in the number of
cycles is substantial. In terms of states and cycles, our tool outperforms Bounded Synthesis
on half of the benchmarks and it outperforms Acacia+ on all
benchmarks. 

The results of Acacia+ show that the number of cycles is indeed an
explosive factor. However, they also show that this explosion can be
avoided effectively.

%% file: conclusions.tex
We have introduced the Bounded Cycle Synthesis problem, where we limit
the number of cycles in an implementation synthesized from an LTL
specification.  Our solution is based on the construction of a witness
structure that limits the number of cycles. The existence of such a
witness can be encoded as a SAT problem. Our experience in applying
Bounded Cycle Synthesis to the synthesis of the AMBA bus arbiter shows
that the approach leads to significantly better
implementations. Furthermore, the performance of our prototype
implementation is suffient to synthesize the components (in a natural
decomposition of the specification) in reasonable time.

Both Bounded Synthesis and Bounded Cycle Synthesis can be seen as the
introduction of structure into the space of implementations. Bounded
Synthesis structures the implementations according to the number of
states, Bounded Cycle Synthesis additionally according to the number
of cycles.  The double exponential blow-up between the size of the
specification and the number of states, and the triple exponential
blow-up between the size and the number of cycles indicate that, while
both parameters provide a fine-grained structure, the number of cycles
may even be the superior parameter.  Formalizing this intuition and
finding other useful parameters is a challenge for future work.

Our method does not lead to a synthesis algorithm in the classical
sense, where just a specification is given and an implementation or an
unsatisfiability result is returned. In our setting, the bounds are
part of the input, and have to be determined beforehand.
In Bounded Synthesis, the bound is usually eliminated by increasing the
bound incrementally. With multiple bounds, the choice which parameter
to increase becomes non-obvious. Finding a good strategy for this problem is a
challenge on its own and beyond the scope of this paper. We
leave it open for future research.

\goodbreak

%% file: appendix.tex
\section{SCC Encoding}
\label{apx:encoding}

In the following, we describe the SAT encoding to guess the SCC
annotations $ \mathcal{F}_{SCC}(n) $ for each sub-graph of a given
graph $ G $, induced by the bound $ n \in \nats $.  Therefore, we fix
a vertex~$ v $ in each SCC and guess two spanning trees, rooted in
$ v $, with the second one being inverted, i.e., the edges lead back
to the root. This ensures, that from the vertex~$ v $ each other
vertex is reachable and from each other vertex, $ v $ is
reachable. The corresponding spanning trees then are witnesses for the
guessed SCCs. To verify that we guessed maximal SSCs, we finally
enforce that the DAG of all SCCs is totally ordered and that edges
have to follow that order.

\medskip

\noindent Let $ T $ be some set with $ \size{T} = n $. We introduce
the following variables for each sub-graph $ 0 < k \leq n $:

\begin{itemize}

\item $ \forwardedge{k}{t}{t'} $ for all $ t,t' \in T $, for the
  edges of the first spanning tree.

\item $ \backwardedge{k}{t}{t'} $ for all $ t,t' \in T $, for the
  edges of the second spanning tree.

\item $ \frankd{k}{t}{i} $ for all $ t \in T $ and
  $ 0 < i \leq \log n $, denoting a ranking function that measures the
  distance from the root of the forward spanning tree.

\item $ \brankd{k}{t}{i} $ for all $ t \in T $ and
  $ 0 < i \leq \log n $, denoting a ranking function that measures the
  distance to the root of the backward spanning tree.

\end{itemize}

\noindent We guss and verify the SCC annotation according to the
following constraints:

\begin{itemize}

\item The SCCs are totally ordered:
  \begin{equation*}
    \bigwedge\limits_{t,t' \in T} \edge{t}{t'} \rightarrow \scc{k}{t} \leq \scc{k}{t'}
  \end{equation*}

\item Each SCC has an SCC root, annotated with the smallest ranking:
  \begin{equation*}
    \bigwedge\limits_{0 < i \leq n} \Big(\big(\bigvee\limits_{t \in T} \scc{k}{t} = i\big) \rightarrow \big(\bigvee\limits_{t \in T} (\scc{k}{t} = i \wedge \frank{k}{t} = 0)\big)\Big) 
  \end{equation*}

\item An SCC root is unique:
  \begin{equation*}
    \bigwedge\limits_{t,t' \in T, \, t \neq t'} \neg \big( \scc{k}{t} = \scc{k}{t'} \wedge \frank{k}{t} = 0 \wedge \frank{k}{t'} = 0 \big)
  \end{equation*}

\item The root is the same according to both rankings:
  \begin{equation*}
    \bigwedge\limits_{t \in T} \frank{k}{t} = 0 \leftrightarrow \brank{k}{t} = 0 
  \end{equation*}

\item SCC roots do not have incoming forward edges nor outgoing
  backward edges:
  \begin{equation*}
    \bigwedge\limits_{t' \in T} \frank{k}{t'} = 0 \rightarrow \bigwedge\limits_{t \in T} \neg \forwardedge{k}{t}{t'} \wedge \neg \backwardedge{k}{t}{t'} 
  \end{equation*}

\item All non-roots have exactely one incoming forward edge:
  \begin{equation*}
    \bigwedge\limits_{t' \in T} \frank{k}{t'} \neq 0 \rightarrow \textit{exactlyOne}(\set{ \forwardedge{k}{t}{t'} \mid t \in T })
  \end{equation*}

\item All non-roots have exactely one outgoing backward edge:
  \begin{equation*}
    \bigwedge\limits_{t' \in T} \brank{k}{t'} \neq 0 \rightarrow \textit{exactlyOne}(\set{ \backwardedge{k}{t'}{t} \mid t \in T })
  \end{equation*}

\item Forward edges preserve the ranking:
  \begin{equation*}
    \bigwedge\limits_{t,t' \in T} \forwardedge{k}{t}{t'} \rightarrow \frank{k}{t} < \frank{k}{t'}
  \end{equation*}

\item Backward edges preserve the ranking:
  \begin{equation*}
    \bigwedge\limits_{t,t' \in T} \backwardedge{k}{t}{t'} \rightarrow \brank{k}{t} > \brank{k}{t'}
  \end{equation*}

\item Only edges of the same SCC can be connected by a forward edge:
  \begin{equation*}
    \bigwedge\limits_{t,t' \in T} \forwardedge{k}{t}{t'} \rightarrow \edge{t}{t'} \wedge \scc{k}{t} = \scc{k}{t'} 
  \end{equation*}

\item Only edges of the same SCC can be connected by a backward edge:
  \begin{equation*}
    \bigwedge\limits_{t,t' \in T} \backwardedge{k}{t}{t'} \rightarrow \edge{t}{t'} \wedge \scc{k}{t} = \scc{k}{t'} 
  \end{equation*}

\end{itemize}

\noindent The resulting formula is quadratic in $ n $ for each
$ 0 < k \leq n $ and consists of $ n^{2} $ many variables. Hence, the
overall formula~$ \curlyF_{SCC}(n) $ consists of at most $ n^{3} $
many variables and $ \size{\curlyF_{SCC}(n)} \in \bigo(n^{3}) $.